\documentclass[a4paper,UKenglish]{lipics-report}

\usepackage{stackrel}
 \usepackage{wrapfig}
\usepackage{verbatim}
\usepackage{color}
\usepackage{xspace}
\usepackage{epsfig}
\usepackage{savesym}
\savesymbol{iint}
\savesymbol{iiint}
\usepackage{amsmath,amssymb}
\usepackage{wasysym}
\usepackage{pstricks,graphicx}
\usepackage{stmaryrd}
\usepackage{enumitem}

\usepackage{thm-restate}
\usepackage{pgf}
\usepackage{qtree}

\usepackage{tikz}
\usetikzlibrary{fit} 
\usetikzlibrary{backgrounds} 
\usetikzlibrary{patterns,calc}
\usetikzlibrary{decorations.pathreplacing}
\usetikzlibrary{positioning}
\usetikzlibrary{decorations.pathmorphing}
\usetikzlibrary{decorations.markings}
\tikzstyle{background}=[rectangle,fill=gray!10, inner sep=0.1cm, rounded corners=0mm]
\usepgflibrary{shapes}
\usetikzlibrary{snakes,automata,arrows}
\usetikzlibrary{shadows}
\tikzstyle{nloc}=[draw, text badly centered, rectangle, rounded corners, minimum size=2em,inner sep=0.5em]
\tikzstyle{background}=[rectangle,fill=gray!10, inner sep=0.1cm, rounded corners=0mm]
\tikzstyle{loc}=[draw,rectangle,minimum size=1.4em,inner sep=0em]
\tikzset{
    gluon/.style={decorate,draw=black,
        decoration={coil,amplitude=1pt, segment length=5pt}} 
}

\tikzset{
    gluonew/.style={decorate,draw=black,
        decoration={coil,amplitude=1pt, segment length=2pt}} 
}

\tikzset{
    gluon1/.style={decorate,draw=black,
        decoration={coil,amplitude=3pt, segment length=3pt}} 
}

\usepackage[pdflatex,recompilepics=auto]{gastex}
\gasset{frame=false,Nw=6,Nh=6,loopdiam=5}

\newcommand{\mypara}[1]{\smallskip\noindent{\textup{\textbf{#1.}}}}
\newcommand{\stackalphabet}{\Lambda}

\newcommand{\AKMV}{\ensuremath{\mathcal{A}^{K,M}_{\textsf{valid}}}\xspace}

\newcommand{\AKMS}{\ensuremath{\mathcal{A}^{K,M}_{\Sys}}\xspace}
\newcommand{\TCW}{\ensuremath{\mathsf{TCW}}\xspace}
\newcommand{\TCWs}{\ensuremath{\mathsf{TCWs}}\xspace}

\newcommand{\Sys}{\ensuremath{\mathcal{S}}\xspace}
\newcommand{\Lang}{\ensuremath{\mathcal{L}}\xspace}
\newcommand{\op}{\ensuremath{\mathsf{op}}\xspace}
\newcommand{\nop}{\ensuremath{\mathsf{nop}}\xspace}
\newcommand{\push}{\ensuremath{\downarrow}\xspace}
\newcommand{\pop}{\ensuremath{\uparrow}\xspace}

\newcommand{\stt}{\tau}

\newcommand{\TT}{\ensuremath{\mathsf{TT}}\xspace}
\newcommand{\TTs}{\ensuremath{\mathsf{TTs}}\xspace}
\newcommand{\kTT}{$K$-\TT}
\newcommand{\kTTs}{$K$-\TTs}

\newcommand{\tcw}{\ensuremath{\mathcal{V}}\xspace}

\newcommand{\add}[3]{\mathop{\mathsf{Add}_{#1,#2}^{#3}}}
\newcommand{\forget}[1]{\mathop{\mathsf{Forget}_#1}}
\newcommand{\rename}[2]{\mathop{\mathsf{Rename}_{#1,#2}}}

\newcommand{\N}{\mathbb{N}}

\newcommand{\sem}[1]{\llbracket #1 \rrbracket}

\newcommand{\matchrel}{\mathbin{\vartriangleright}}
\newcommand{\procrel}{\rightarrow}
\newcommand{\hole}{\dashrightarrow}

 \newcommand{\cL}{{\cal L}}
 
\newcommand{\cI}{{\cal I}} 
 
\newcommand{\cO}{{\cal O}}

\newcommand{\true}{\textsf{tt}} 
\newcommand{\false}{\textsf{ff}} 

\newcommand{\dom}{\ensuremath{\textsf{{\sffamily dom}}}}

\newcommand{\Real}{\mathsf{Real}}

\newcommand{\Act}{\ensuremath{\mathsf{Act}}\xspace}
\newcommand{\Left}{\ensuremath{\mathsf{Left}}\xspace}
\newcommand{\Right}{\ensuremath{\mathsf{Right}}\xspace}

\newcommand{\nxt}{\ensuremath{\mathsf{next}}\xspace}
\newcommand{\prv}{\ensuremath{\mathsf{prev}}\xspace}

\newcommand{\ts}{\ensuremath{\mathsf{ts}}\xspace}
\newcommand{\tsm}{\ensuremath{\mathsf{tsm}}\xspace}
\newcommand{\acc}{\ensuremath{\mathsf{acc}}\xspace}
\newcommand{\ACC}{\ensuremath{\mathsf{ACC}}\xspace}
\newcommand{\Bool}{\ensuremath{\mathbb{B}}\xspace}
\renewcommand{\matchrel}{\curvearrowright}

\newcommand{\pushb}{\mathsf{Push}}
\newcommand{\popb}{\mathsf{Pop}}
\newcommand{\reset}{\mathsf{Reset}}

\newtheorem{claim}[theorem]{Claim}

\usepackage{todonotes}

\newcounter{todocounter}

\begin{document}

\title{Towards an Efficient Tree Automata based technique for Timed Systems
  \footnote{This work was partly supported by UMI-ReLaX, DST-CEFIPRA project AVeRTS and DST-INSPIRE faculty award [IFA12-MA-17].}}
\titlerunning{Towards an Efficient Tree Automata based technique for Timed Systems}

\author[1]{S. Akshay}
\author[2]{Paul Gastin}
\author[1]{S. Krishna}
\author[1]{Ilias Sarkar}
\affil[1]{Dept of CSE, IIT Bombay, India\\
  \texttt{akshayss,krishnas,ilias@cse.iitb.ac.in}}
\affil[2]{LSV, ENS Paris-Saclay, CNRS, France\\
  \texttt{paul.gastin@lsv.fr}}
\authorrunning{S. Akshay, P. Gastin, S. Krishna, I. Sarkar} 

\maketitle

\begin{abstract}
The focus of this paper is the analysis of real-time systems with recursion, through the development of good theoretical techniques which are implementable. Time is modeled using clock variables, and recursion using stacks. Our technique consists of  modeling the behaviours of the timed system as graphs, and interpreting these graphs on tree terms by showing a bound on their tree-width. We then build a tree automaton that accepts exactly those tree terms that describe realizable runs of the timed system. The emptiness of the timed system thus boils down to emptiness of a finite tree automaton that accepts these tree terms. This approach helps us in obtaining an optimal complexity, not just in theory (as done in earlier work e.g.~\cite{concur16}), but also in going towards an efficient implementation of our technique. To do this, we make several improvements in the theory and exploit these to build a first prototype tool that can analyze timed systems with recursion.
\end{abstract}

\section{Introduction}
Development of efficient techniques for the verification of real time systems is a practically relevant problem.  Timed automata \cite{AD94} are a  prominent and well accepted abstraction of timed systems.  The development of this model originally began with highly theoretical results, starting from the $\mathsf{PSPACE}$-decision procedure for the emptiness of timed automata. But later, this theory has led to the development of state of the art and industrial strength tools like UPPAAL \cite{uppaal}. Currently, such tools are being adapted to build prototypes that handle other systems such timed games, stochastic timed systems etc. 
While this helps in analysis of certain systems, there are complicated real life examples that require paradigms like recursion, multi-threaded concurrency and so on.  

For timed systems with recursion, a popular theoretical framework is the model of timed pushdown automata (TPDA). In this model, in addition to clock variables as in timed automata, a stack is used to model recursion. Depending on how clocks and stack operations are integrated, several variants~\cite{BER94}, \cite{lics12}, \cite{TW10}, \cite{LCOY13}, \cite{CL15} have been looked at. For many of these variants, the basic problem of checking emptiness has been shown decidable (and $\mathsf{EXPTIME}$-complete) using different techniques. The proofs in \cite{BER94}, \cite{lics12}, \cite{TW10} work by adapting the technique of region abstraction to untime the stack and obtain a usual untimed pushdown automaton, while \cite{CL15} gives a proof by reasoning with sets of timed atoms. Recently, in~\cite{concur16}, a new proof technique was introduced which modeled the behaviours of the TPDA as graphs with timing constraints and analyzed these infinite collections of time-constrained graphs using tree automata. This approach follows the template which has been explored in depth for various untimed systems in~\cite{MP11}, \cite{CGK12}, \cite{CG14}. The basic idea can be outlined as follows: (1) describe behaviours of the underlying system as graphs, (2) show that this class of graphs has bounded width, (3) either appeal to Courcelle's theorem~\cite{CourcelleBook} by showing that the desired properties are MSO-defineable or explicitly construct a tree-automaton to capture the class of graphs that are the desired behaviours. The work in~\cite{concur16} extends this approach to timed systems, by considering their behaviors as time-constrained words. The main difficulty here is to obtain a tree automaton that accepts only those time-constrained words that are \emph{realizable} via a valid time-stamping.

Despite the amount of theoretical work in this area~\cite{BER94,MP11,CGK12,concur16,lics12,CL15}, none of these algorithms have been implemented to the best of our knowledge. Applying Courcelle's theorem is known to involve a blowup in the complexity (depending on the quantifier-alternation of the MSO formula). The algorithm for checking emptiness in~\cite{concur16} for the timed setting  which directly constructs the tree automaton avoiding the MSO translation also turns out to be unimplementable even for small examples due to the following reasons:  First, it has a pre-processing step where each transition in the underlying automaton is broken into several micro transitions, one for each  constraint that is checked there, and one corresponding to each clock that gets reset on that transition.  This results in a blowup in the size of the automaton. Second, the number of states of the tree automaton that is built to check realizabilty as well as the existence of a run of a system is bounded by $(M\times T)^{\mathcal{O}(K^2)}2^{\mathcal{O}(K^2lg K)}$, where $M$ is one more than the maximal constant used in the given system, $T$ is the number of transitions, and $K=4|X|+6$ is the so-called split-width, where $|X|$ is the number of clocks used. This implies that even for a system that has $1$ clock, $5$ transitions and uses a maximum constant $5$, we have more than $30^{100}$ states.

In this paper, we take the first steps towards an efficient implementation. While we broadly follow the graph and tree-automata based approach (and in particular~\cite{concur16}), our main contribution is to give an efficient technique for analyzing TPDA. This requires several fundamental advances: (i) we avoid the preprocessing step, obtaining a direct bound on tree width for timed automata and TPDA. This is established by playing a \emph{split-game} which decomposes the graph representing behaviours of the timed system into tree terms; by coloring some vertices of the graph and removing certain edges whose endpoints are colored. The minimum number of colors used in a winning strategy is 1 plus the tree-width of the graph. (ii) we develop a new algorithm for building the tree automaton for emptiness, whose complexity is in $\mathsf{ETIME}$, i.e., bounded by $(M\times T)^{3|X|+3}$ with an exponent which is a linear function of the input size (improved from EXPTIME, where the exponent is a polynomial function of the input). Thus, if the system has $1$ clock, $5$ transitions and uses a maximum constant $5$, we have only $\sim 30^6$ states. In particular, our tree-automaton is \emph{strategy-driven}, i.e.,  it manipulates only those tree terms that arise out of a winning strategy of our split-game. As a result of this strategy-guided approach, the number of states of our tree automaton is highly optimized, and an accepting run exactly corresponds to the moves in a winning strategy of our split-game.   (iii) Finally, our algorithm outputs a witness for realizability (and non-emptiness). As a proof-of-concept, we implemented our algorithm and despite the worst-case complexity, in Section \ref{sec:expt}, we discuss optimizations, results and a modeling example where our implementation performs well.

\section{Graphs for behaviors of timed systems}
\label{sec-prelim}

We fix an alphabet $\Sigma$ and use $\Sigma_{\varepsilon}$ to denote $\Sigma
\cup \{\varepsilon\}$, where $\varepsilon$ is the silent action.  We also fix a
finite set of intervals $\cI$ with bounds in $\mathbb{N}\cup\{\infty\}$. For a set $S$, we use ${\leq}\subseteq {S\times S}$ to denote a partial or total
order on $S$.  For any $x,y\in S$, we write $x<y$ if $x\leq y$ and $x\neq y$,
and $x\lessdot y$ if $x<y$ and there does not exist $z\in S$ such that $x<z<y$.

\subsection{Abstractions of timed behaviors}
\label{tcwords}
\begin{definition}
\label{def:tcw}
  A \emph{word with timing constraints} (\TCW) over $(\Sigma,\cI)$ is a
  structure \\ $\tcw{=}(V,\procrel,\lambda,(\curvearrowright^I)_{I\in\cI})$ where
  $V$ is a finite set of vertices or positions, $\lambda\colon V\to
  \Sigma_{\varepsilon}$ labels each position, the reflexive transitive closure
  ${\leq}={\procrel}^*$ is a total order on $V$ and ${\procrel}={\lessdot}$ is
  the successor relation, while ${\curvearrowright}^I\subseteq{<}$ connects
  pairs of positions carrying a timing constraint, given by an interval in $I
  \in \cI$.
  A \TCW $\tcw{=}(V,\procrel,\lambda,(\curvearrowright^I)_{I\in\cI})$ is
called \emph{realizable} if there exists a timestamp map $\ts \colon V\to\mathbb{R}_+$ such that $\ts(i)\leq\ts(j)$ for all $i\leq j$ (time is non-decreasing) and
$\ts(j)-\ts(i)\in I$ for all $i\matchrel^{I}j$ (timing constraints are
satisfied).
\end{definition} 

An example of a \TCW is given in Figure~\ref{fig2} (right), with positions $0, 1, 2, 3$ labelled by $\Sigma=\{a,b,c\}$. Curved edges decorated with intervals connect positions related by $\curvearrowright^I$, while straight edges define the successor relation $\procrel$. This \TCW is realizable by the sequence of timestamps $0,0.9,2.89,3.1$ but not by $0, 0.9, 2.99, 3.1$.  We let $\Real(\Sigma,\cI)$ be the set of \TCWs over $(\Sigma,\cI)$ which are realizable.

\subsection{TPDA and their semantics as \TCWs}
\label{sec:tpda-sem}  Dense-timed pushdown automata (TPDA), introduced in \cite{lics12}, are an extension of timed automata, and operate on a finite set of real-valued clocks
and a stack which holds symbols with their ages.  The age of a symbol 
represents the time elapsed since it was pushed onto the stack.
Formally, a TPDA $\Sys$ is a tuple $(S, s_{0}, \Sigma, \stackalphabet, \Delta, X, F)$ where $S$ is a finite set of states, $s_{0} \in S$ is the initial state,
$\Sigma$, $\stackalphabet$, are respectively finite sets of input, stack symbols, $\Delta$ is a finite set of transitions, $X$ is a finite set of real-valued
variables called clocks, $F\subseteq S$ are final states.  A transition $t \in
\Delta$ is a tuple $(s, \gamma, a, \op, R, s')$ where $s, s' \in S$, $a\in
\Sigma$, $\gamma$ is a finite conjunction of atomic formulae of the kind $x\in
I$ for $x \in X$ and $I \in \mathcal{I}$, $R \subseteq X$ are the clocks reset,
$\op$ is one of the following stack operations:
\begin{enumerate}
  \item \nop does not change the contents of the stack,
  \item $\push_c$, $c \in \stackalphabet$ is a push operation that adds $c$ on top of the stack, with age 0.
  \item $\pop^I_c$, $c \in \stackalphabet$ is a stack symbol and $I \in \mathcal{I}$ is an interval, is a pop operation 
  that removes the top most symbol of the stack provided it is a $c$ with age in the interval $I$. 
\end{enumerate}
Timed automata (TA) can be seen as TPDA using \nop operations only.  This
definition of TPDA is equivalent to the one in~\cite{lics12}, but allows
checking conjunctive constraints and stack operations together.  In \cite{CL15},
it is shown that TPDA of~\cite{lics12} are expressively equivalent to timed
automata with an untimed stack.  As our technique is oblivious to whether the stack is timed or not, we focus on the syntactically more succinct model TPDA with a timed stack. 

Next, we define the semantics of a TPDA in terms of \TCWs.
\begin{definition}
\label{def:semantics}
 A \TCW
$\tcw=(V,\procrel,\lambda,(\curvearrowright^I)_{I\in\cI})$ is said to be \emph{generated or accepted} by a TPDA $\Sys$ if there is an accepting abstract run
$\rho=(s_0,\gamma_1,a_1,\op_1,R_1,s_1)$ $(s_1,\gamma_2,a_2,\op_2,R_2,s_2)\cdots$
$(s_{n-1},\gamma_n,a_n,\op_n,R_n,s_n)$ of $\Sys$ such that, $s_n\in F$ and
\begin{itemize}
  \item the sequence of push-pop operations is well-nested: in each prefix
  $\op_1\cdots\op_k$ with $1\leq k\leq n$, number of pops is at most number of
  pushes, and in the full sequence $\op_1\cdots\op_n$, they are equal; and
  
  \item $V=\{0,1,\ldots,n\}$ with $\lambda(0)=\varepsilon$ and 
  $\lambda(i)=a_i$ for all $1\leq i\leq n$ and $0\procrel 1\procrel \cdots 
  \procrel n$ and, for all $I\in\cI$, the relation $\matchrel^{I}$ is the set of 
  pairs $(i,j)$ with $0\leq i<j\leq n$ such that
  \begin{itemize}
    \item either for some $x\in X$ we have $x\in R_i$ (assuming $R_0=X$) and
    $x\in I$ is a conjunct of $\gamma_j$ and $x\notin R_k$ for all $i<k<j$,
  
    \item or $\op_i={\push}_b$ is a push and $\op_j={\pop}_b^{I}$ is the
    matching pop (same number of pushes and pops in $\op_{i+1}\cdots\op_{j-1}$).
  \end{itemize}
\end{itemize}
\end{definition}
We denote by $\TCW(\Sys)$ the set of \TCWs generated by $\Sys$.
The non-emptiness  problem for the TPDA \Sys amounts to asking whether some \TCW generated by \Sys is  realizable, i.e., whether $\TCW(\Sys)\cap\Real(\Sigma, \cI) \neq\emptyset$.
The $\TCW$ semantics of timed automata (TA) can be obtained from the above
discussion by just ignoring the stack components (using \nop operations only).
Figure~\ref{fig2} depicts a simple example of a timed automaton and a $\TCW$ generated by it.
\begin{figure}[t]
\includegraphics[scale=.46,page=1]{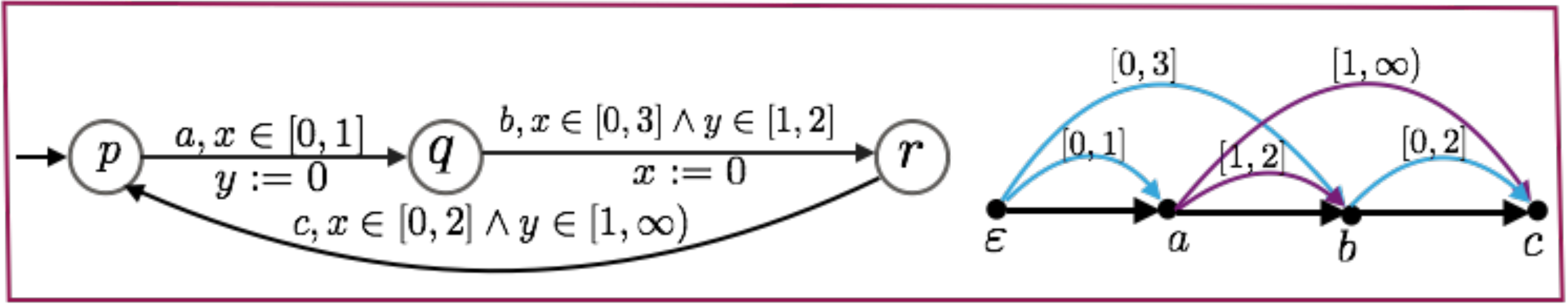}
\caption{A timed automaton and  a \TCW capturing a run}
\label{fig2}
\end{figure}

\begin{remark}
  The classical semantics of timed systems is given in terms of timed words.  A
  \emph{timed word} is a sequence $w=(a_1,t_1)\cdots(a_n,t_n)$ with
  $a_1,\ldots,a_n\in\Sigma$ and $(t_i)_{1\leq i\leq n}$ is a non-decreasing
  sequence of values in $\mathbb{R}_{+}$.  A \emph{realization} of a \TCW
  $\tcw=(V,\procrel,\lambda,(\curvearrowright^I)_{I\in\cI})\in\TCW(\Sys)$ with
  $V=\{0,1,\ldots,n\}$ is a timed word
  $w=(\lambda(1),\ts(1))\ldots(\lambda(n),\ts(n))$ where the timestamp map
  $\ts\colon V\to\mathbb{R}_{+}$ (with $\ts(0)=0)$ is non decreasing and
  satisfies all timing constraints of $\tcw$.  For example, the timed word
  $(a,0.9)(b,2.89)(c,3.1)$ is a realization of the \TCW in Figure~\ref{fig2}
  while $(a,0.9)(b,2.99)(c,3.1)$ is not. It is not difficult to check that the 
  language $\cL(\Sys)$ of timed words accepted by \Sys with the classical 
  semantics is precisely the set of realizations of \TCWs in $\TCW(\Sys)$.
  Therefore, $\cL(\Sys)=\emptyset$ iff $\TCW(\Sys)\cap\Real(\Sigma, \cI)=\emptyset$.
\end{remark}

We now identify some important properties satisfied by \TCWs generated from a
TPDA. Let $\tcw=(V,\procrel,\lambda,(\curvearrowright^{I})_{I\in\cI})$ be a
\TCW.
The matching relation $(\curvearrowright^{I})_{I\in\cI}$ is used in two 
contexts: (i) while connecting a clock reset point (say for clock $x$) to a
point where a guard of the form $x\in I$ is checked, and (ii) while connecting a
point where a push was made to its corresponding pop, where the age of the
topmost stack symbol is checked to be in interval $I$.
We use the notations  $\curvearrowright^{x\in I}$ 
and $\curvearrowright^{s\in I}$ to 
denote the matching relation  $\curvearrowright^{I}$ 
 corresponding to  a clock-reset-check as well as 
 push on stack-check respectively.
We say that \tcw is \emph{well timed} w.r.t.\ a set of clocks $X$ and a stack
$s$ if for each interval $I\in\cI$ the matching relation $\curvearrowright^{I}$
can be partitioned as ${\curvearrowright}^{I}={\curvearrowright}^{s\in I}\uplus
\biguplus_{x\in X}{\curvearrowright}^{x\in I}$ where
\begin{enumerate}[label=$(\mathsf{T}_{\arabic*})$,ref=$\mathsf{T}_{\arabic*}$]
  \item\label{item:T1}
  the stack relation ${\curvearrowright}^s=\bigcup_{I\in\cI}{\matchrel}^{s\in I}$
  corresponds to the matching push-pop events, hence it is well-nested: for all
  $i\curvearrowright^s j$ and $i'\curvearrowright^s j'$, if $i<i'<j$ then
  $j'<j$.
  
  \item\label{item:T2} For each $x\in X$, the clock relation
  ${\curvearrowright}^x=\bigcup_{I\in\cI}{\matchrel}^{x\in I}$ corresponds to
  the timing constraints for clock $x$ and respects the last reset condition:
  for all $i\curvearrowright^x j$ and $i'\curvearrowright^x j'$, if $i<i'$, then
  $j \leq i'$.  See Figure \ref{fig2} for example, where $0\curvearrowright^x
  2$ and $2 \curvearrowright^x 3$.
\end{enumerate}
It is then easy to check that \TCWs defined by a TPDA with set of clocks $X$
are well-timed for the set of clocks $X$, i.e., satisfy the
properties above.  We
obtain the same for TA by just ignoring the stack edges, i.e., (\ref{item:T1})
above.

\section{Tree-Width for Timed Systems}
\label{sec:stw-tpda}
In this section, we discuss tree-algebra by introducing the basic terms, the
operations on terms, their syntax and semantics.  This will help us in analyzing
the graphs obtained in the previous section using tree-terms, and establishing a
bound on the tree-width.  
\begin{figure}[t]
\includegraphics[scale=.22,page=5]{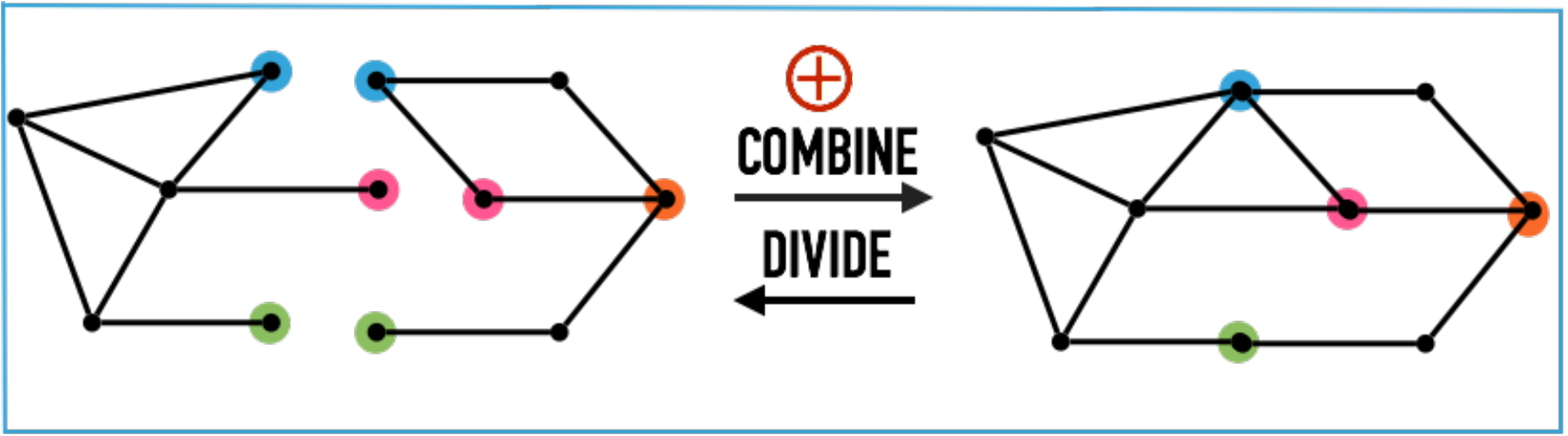}
\hfill
\includegraphics[scale=.22,page=6]{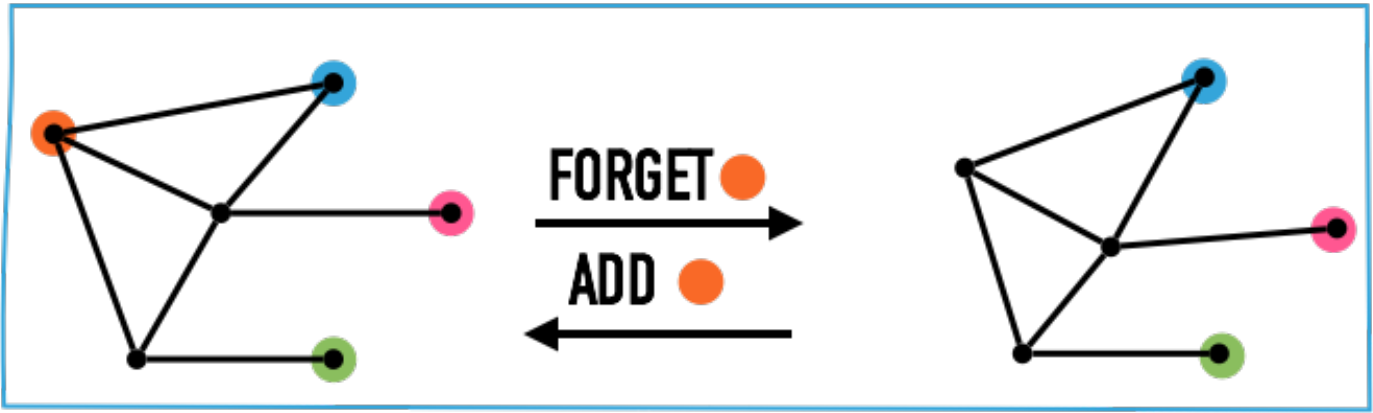}
\caption{Operations on colored graphs.}
\protect\label{fig:combine}
\end{figure}
We introduce tree terms $\TTs$ from Courcelle
\cite{CourcelleBook} and their semantics as graphs
which are both vertex-labeled
and edge-labeled.  Let $\Sigma$ be a set of vertex labels and let $\Xi$ be a set of edge
labels. 
Let $K \in \mathbb{N}$.  
 The syntax of $K$-tree terms \kTTs over $(\Sigma,\Xi)$
is given by 

\noindent\hfil$
\stt ::= (a,i) \mid {(a,i)}\xi{(b,j)}\mid \forget{i}\stt
\mid \rename{i}{j}\stt \mid \stt \oplus \stt
$

\noindent%
where $i,j\in\{1,2,\dots,K\}$ are colors ($i \neq j$), $a,b\in\Sigma$ are
vertex labels and $\xi\in\Xi$ is an edge label.  The semantics of a \kTT $\stt$
is a colored graph $\sem\stt=(G_\stt,\chi_\stt)$ where $G_\stt=(V,E)$ is a graph and
$\chi_\stt\colon \{1,2,\dots,K\}\to V$ is a partial injective function assigning a color to
some vertices of $G_\stt$. Note that any color in $\{1,2,\dots,K\}$ is assigned 
to at most one vertex of $G_\stt$. 

The atomic term $(a,i)$ is a single vertex colored $i$ and labeled $a$ and the
atomic term ${(a,i)}\xi{(b,j)}$ represents a $\xi$-labeled edge
between two vertices colored $i, j$ and labeled $a,b$ respectively.  Given a
tree term $\tau$, $\forget{i}(\tau)$ forgets the color $i$ from a
node colored $i$, leaving it uncolored.  The operation $\rename{i}{j}(\stt)$
renames the color $i$ of a node to color $j$, provided no nodes are already
colored $j$. Since any color appears at most once 
in $G_\stt$, the operations  $\forget{i}(\tau)$ and $\rename{i}{j}(\stt)$
are  deterministic, when  colors $i, j$, are fixed. 
  Finally, the operation $\stt_1 \oplus \stt_2$ (read as combine)
combines two terms $\stt_1, \stt_2$ by fusing the nodes of $\stt_1, \stt_2$
which have the same color. 
See Figure~\ref{fig:combine}.

The tree-width of a graph $G$ is defined as the least $K$ such that $G=G_{\tau}$
for some \TT $\tau$ using $K+1$ colors. Let $\mathsf{TW}_K$ denote the set of all graphs having tree width at most $K$.
For \TCWs, we have successor edges $\rightarrow$ and matching edges
$\curvearrowright^{I}$ where $I\in\cI$ is an interval.  Hence, the set of edge
labels is $\Xi_\cI=\{{\procrel}\}\cup\{{\matchrel^{I}}\mid I\in\cI\}$ and we 
use \TTs over $(\Sigma,\Xi_\cI)$. An example is given in Appendix~\ref{app:TT-example}.

\subsection{TCWs and Games}\label{sec:game} 
We find it convenient to prove that \TCWs have bounded tree-width by playing a
game, whose game positions are \TCWs in which some successor edges may have been
cut, i.e., are missing.  Such \TCWs, where some successor edges may be missing, are called 
split-\TCWs. A split-\TCW which is a  connected graph is called a connected split-\TCW, while 
 a split-\TCW which is a disconnected graph, is called a disconnected  
split-\TCW.  For example, \includegraphics[scale=.25,page=38]{all-figs} is a connected split-\TCW, while 
\includegraphics[scale=.25,page=39]{all-figs} is a disconnected split-\TCW 
consisting of two connected split \TCWs, namely \includegraphics[scale=.25,page=40]{all-figs}
and \includegraphics[scale=.25,page=41]{all-figs}.

A \TCW is atomic if it is denoted by an atomic term ($(a,i)$ or $(a,i)
\rightarrow (b,j)$ or $(a,i) \curvearrowright^{I} (b,j)$).
\emph{The split-game} is a two player turn based game
$\mathcal{G}=(\mathsf{Pos}_\exists\uplus \mathsf{Pos}_\forall,\mathsf{Moves})$ where Eve's set of game positions
$\mathsf{Pos}_\exists$ consists of all connected (wrt.\ ${\procrel}\cup{\curvearrowright}$)
split-\TCWs and Adam's set of game positions $\mathsf{Pos}_\forall$ consists of dis-connected
split-\TCWs. 
Eve's moves consist of adding colors to the vertices of the split-\TCW, and 
dividing the split-\TCW. For example, if we have the connected split-\TCW 
\includegraphics[scale=.25,page=38]{all-figs}, and Eve colors two nodes (we use shapes 
in place of colors for better visibility) 
we obtain \includegraphics[scale=.25,page=42]{all-figs}. 
This graph can be  divided  obtaining the  disconnected graph \includegraphics[scale=.25,page=43]{all-figs} and 
\includegraphics[scale=.25,page=44]{all-figs}. 
 As a result, we obtain the connected parts  \includegraphics[scale=.25,page=45]{all-figs} and \includegraphics[scale=.25,page=46]{all-figs} and 
  \includegraphics[scale=.25,page=44]{all-figs}. 
  Now Adam's choices are on this disconnected split-\TCW 
and he can choose either of the above three connected split-\TCWs
to continue the game. 
Thus,  divide is the reverse of the combine operation $\oplus$.   
  Adam's moves amount to choosing a connected component of the
split-\TCW. Eve has to continue coloring and dividing on the 
connected split-\TCW chosen by Adam. 
Atomic split-\TCWs are terminal positions in the game: neither Eve
nor Adam can move from an atomic split-\TCW.
A play on a split-\TCW $\tcw$ is a path in $\mathcal{G}$ starting from
$\tcw$ and leading to an atomic split-\TCW. The cost of the play is
the maximum width (number of colors-1) of any split-\TCW encountered in the path. 
In our example above, \includegraphics[scale=.25,page=44]{all-figs} is already an atomic 
split-\TCW.
If Adam chooses any of the other two, it is easy to see that Eve has a strategy
using at most 2 colors in any of the split-\TCWs that will be obtained till
termination.
The \emph{cost} of a strategy $\sigma$ for Eve from a split-\TCW $\tcw$ is the maximal cost of the
plays starting from \tcw and following strategy $\sigma$.
The \emph{tree-width} of a (split-)\TCW $\tcw$ is the minimal cost
of Eve's (positional) strategies starting from \tcw.
Let $\TCW_K$  denote the set of \TCWs with tree-width
bounded by $K$. 

A \emph{block} in a split-\TCW is a maximal set of
points of $V$ connected by $\procrel$.  For example, the split-\TCW
\includegraphics[scale=.25,page=21]{all-figs}
has one non-trivial block
\includegraphics[scale=.25,page=14]{all-figs}
and one trivial block \includegraphics[scale=.25,page=13]{all-figs}.
Points that are not left or right endpoints of blocks of \tcw are called
internal.

\mypara{The Bound} %
We show that we can find a $K$ such that all the 
behaviors of the given timed system have tree-width bounded by $K$.
\begin{theorem}\label{thm:sw-tpda}
  Given a timed system \Sys using a set of clocks $X$, all graphs in its  \TCW
  language have tree-width bounded by $K$, i.e., $\TCW(\Sys)\subseteq \TCW_K$,
  where
  \begin{enumerate}[nosep]
    \item\label{item:TA-sw} $K=|X|+1$ if \Sys is a timed automaton,
    \item\label{item:TPDA-sw} $K=3|X|+2$ if \Sys is a timed pushdown automaton.
  \end{enumerate}
\end{theorem}

The following lemma completes the proof of
Theorem~\ref{thm:sw-tpda}~\eqref{item:TPDA-sw}.
\begin{lemma}\label{TPDA-bound}
  The tree-width of a well-timed  \TCW is bounded by $3|X|+2$.
\end{lemma}
We prove this by playing the ``split game'' between \emph{Adam} and \emph{Eve}
in which \emph{Eve} has a strategy to disconnect the word without introducing
more than $3|X|+3$ colors.  \emph{Eve}'s strategy processes the word from right
to left.  Starting from any \TCW, Eve colors the end points of the \TCW, as well
as the last reset points (from the right end) corresponding to each clock.  Here
she uses at most $|X|+2$ colors.  On top of this, depending on the last point,
we have different cases. A detailed proof can be seen in Appendix \ref{app:sw}, while we give a sketch here. 

If the last point is the target of a  $\curvearrowright^{x}$  edge for some clock
 $x$, then Eve simply removes the clock edge, since both the source and target points 
  of this edge are colored. 
We only discuss in some detail the case when the last point is the target of a $\curvearrowright^{s}$  edge, and the source 
of this edge is an internal point in the non-trivial block. 
Figure \ref{game-main} illustrates this case.  

To keep a bound on the number of colors needed, Eve divides the \TCW as follows:
\begin{itemize}
\item First Eve adds a color to the  source of the stack edge
\item If there are any  clock edges crossing this stack edge,  
 Eve adds colors to the corresponding reset points. Note that this results in adding atmost $|X|$ colors.
\item Eve disconnects the \TCW into two parts, such that the right part $\tcw_2$
consists of one non-trivial block whose end points are the source and target points of the stack edge, 
and also contains to the left of this block, atmost $|X|$ trivial blocks. 
Each of these trivial blocks are 
the reset points of those clock edges which cross over. 
 The left part $\tcw_1$ is a \TCW consisting of all points to the left 
of the source of the stack edge, and has all remaining edges other than the clock edges which have crossed over.     
Adam can now continue the game choosing $\tcw_1$ or $\tcw_2$.  Note that in one of the words so obtained, the stack edge completely spans the non-trivial block, and 
can be easily removed.   
\end{itemize}

\begin{wrapfigure}[11]{o}[0pt]{0.5\textwidth}
  \includegraphics[scale=.22,page=22]{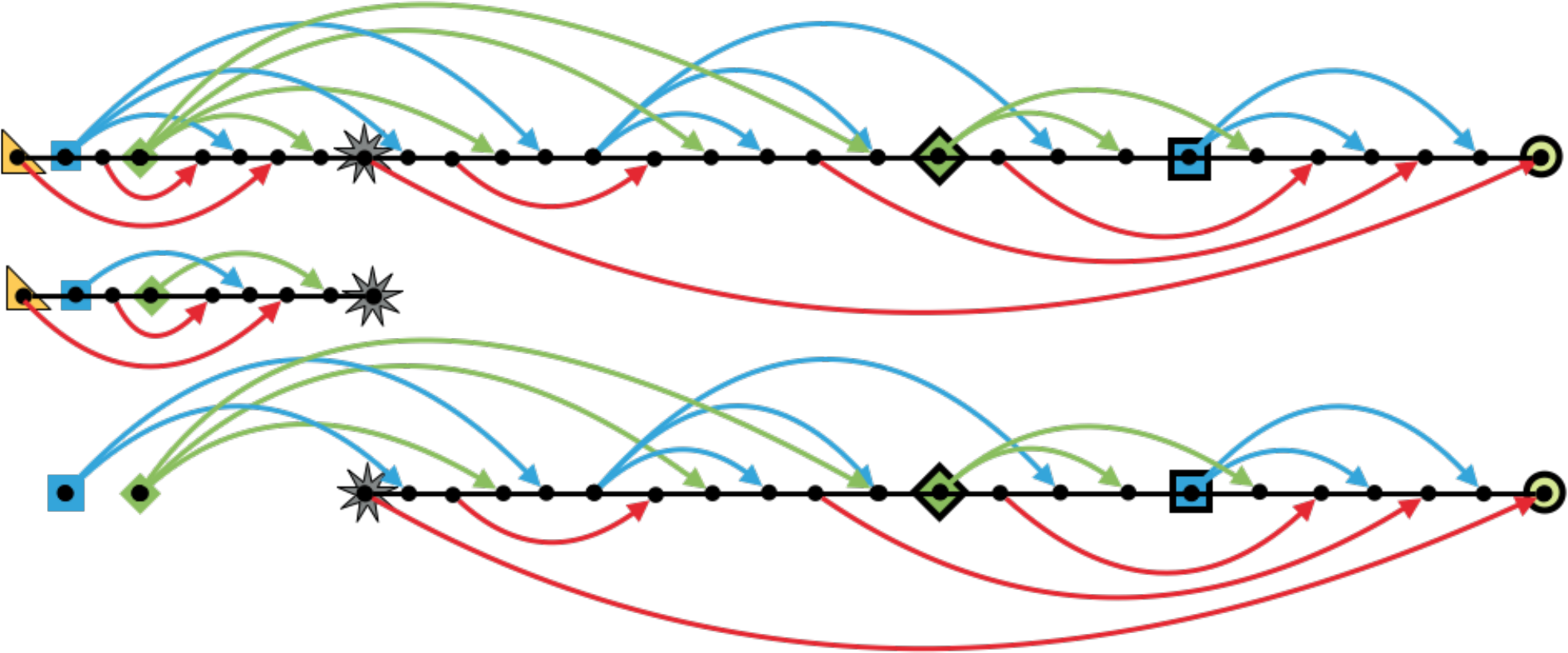}
  \caption{The last point is the target of a $\curvearrowright^s$ (top figure). After the split, we obtain 
  the words $\tcw_1$ (the middle one) and $\tcw_2$ (the bottom one).}
  \label{game-main}
\end{wrapfigure}

\noindent{\emph{Invariants and bound on tree-width}}.  We now discuss some
invariants on the structure of the split-\TCWs as we play the game using the
above strategy.  
\begin{itemize}
\item[(I1)] We have  $\leq |X|$ colored trivial blocks to the left of the only non-trivial block,  
\item[(I2)] The last reset node of each clock on the non-trivial block is colored,
\item[(I3)] The end points of the non-trivial block are colored.
\end{itemize}
To maintain the above invariants, we need $|X|+1$ extra colors than the at most
$2|X|+2$ mentioned above.  This proves that the tree-width of a TPDA with set of
clocks $X$ is bounded by $3|X|+2$. If the underlying system is a timed automaton, then 
we have a single non-trivial block in the game at any point of time. There are no 
 trivial blocks, unlike the TPDA, due to the absence  of stack edges. This results in using only $\leq |X|+2$ colors at any point of time, 
where $|X|$ colors are needed to color the last reset points of the clocks in the block, and the remaining two colors are used to color the right and left end points  
 of the block.

\begin{gpicture}[name=tau1,ignore]
  \gasset{Nframe=n,Nw=2,Nh=2,Nfill=y}
  \node[Nfill=n](0)(3,0){}
  \nodelabel[ExtNL=y,NLdist=1,NLangle=-90](0){$\chi$}
  \nodelabel[ExtNL=y,NLdist=5,NLangle=-90](0){$\ts$}
  \node(1)(08,0){}
  \nodelabel[ExtNL=y,NLdist=1,NLangle=-90](1){1}
  \nodelabel[ExtNL=y,NLdist=5,NLangle=-90](1){0}
  \node(3)(16,0){}
  \nodelabel[ExtNL=y,NLdist=1,NLangle=-90](3){3}
  \nodelabel[ExtNL=y,NLdist=5,NLangle=-90](3){5}
  \node(4)(24,0){}
  \nodelabel[ExtNL=y,NLdist=1,NLangle=-90](4){4}
  \nodelabel[ExtNL=y,NLdist=5,NLangle=-90](4){6}
  \node(5)(32,0){}
  \nodelabel[ExtNL=y,NLdist=1,NLangle=-90](5){5}
  \nodelabel[ExtNL=y,NLdist=5,NLangle=-90](5){8}
  \drawedge[dash={1}0](1,3){}
  \drawedge(3,4){}\drawedge(4,5){}
  \drawedge[curvedepth=4,ELpos=45,ELdist=0](1,4){\tiny$[2,\infty]$}
  \drawedge[curvedepth=8,ELpos=50,ELdist=0](1,5){\tiny$[3,\infty]$}
  \drawedge[curvedepth=3,ELpos=40,ELdist=0](3,5){\tiny$[1,3]$}
\end{gpicture}

\begin{gpicture}[name=state1,ignore]
  \gasset{Nframe=n,Nw=2.5,Nh=3.5,Nadjustdist=0.5}
  \node(0)(3,0){$P$}
  \nodelabel[ExtNL=y,NLdist=1,NLangle=-90](0){$\tsm$}
  \node(1)(08,0){1}
  \nodelabel[ExtNL=y,NLdist=1,NLangle=-90](1){0}
  \node[Nframe=y,Nw=3.5,Nh=3.5](3)(16,0){3}
  \nodelabel[ExtNL=y,NLdist=1,NLangle=-90](3){1}
  \node(4)(24,0){4}
  \nodelabel[ExtNL=y,NLdist=1,NLangle=-90](4){2}
  \node(5)(32,0){5}
  \nodelabel[ExtNL=y,NLdist=1,NLangle=-90](5){0}
  \drawedge[linegray=1,ELside=r,ELdist=3](1,3){$\scriptstyle\neg\acc$}
  \drawedge[linegray=1,ELside=r,ELdist=3](3,4){$\scriptstyle\acc$}
  \drawedge[linegray=1,ELside=r,ELdist=3](4,5){$\scriptstyle\acc$}
\end{gpicture}

\begin{gpicture}[name=tau2,ignore]
  \gasset{Nframe=n,Nw=2,Nh=2,Nfill=y}
  \node[Nfill=n](0)(11,0){}
  \nodelabel[ExtNL=y,NLdist=1,NLangle=-90](0){$\chi$}
  \nodelabel[ExtNL=y,NLdist=5,NLangle=-90](0){$\ts$}
  \node(2)(16,0){}
  \nodelabel[ExtNL=y,NLdist=1,NLangle=-90](2){2}
  \nodelabel[ExtNL=y,NLdist=5,NLangle=-90](2){3}
  \node(4)(24,0){}
  \nodelabel[ExtNL=y,NLdist=1,NLangle=-90](4){3}
  \nodelabel[ExtNL=y,NLdist=5,NLangle=-90](4){6}
  \node(5)(32,0){}
  \nodelabel[ExtNL=y,NLdist=1,NLangle=-90](5){4}
  \nodelabel[ExtNL=y,NLdist=5,NLangle=-90](5){8}
  \node(6)(40,0){}
  \nodelabel[ExtNL=y,NLdist=1,NLangle=-90](6){5}
  \nodelabel[ExtNL=y,NLdist=5,NLangle=-90](6){8}
  \node(7)(48,0){}
  \nodelabel[ExtNL=y,NLdist=1,NLangle=-90](7){6}
  \nodelabel[ExtNL=y,NLdist=5,NLangle=-90](7){11}
  \drawedge[dash={1}0](2,4){}\drawedge[dash={1}0](4,5){}
  \drawedge(5,6){}\drawedge(6,7){}
  \drawedge[curvedepth=8,ELpos=50,ELdist=0](2,7){\tiny$[3,\infty]$}
  \drawedge[curvedepth=3,ELpos=50,ELdist=0](4,6){\tiny$[0,2]$}
  \drawedge[curvedepth=3,ELpos=45,ELdist=0](5,7){\tiny$[1,3]$}
\end{gpicture}

\begin{gpicture}[name=state2,ignore]
  \gasset{Nframe=n,Nw=2.5,Nh=3.5,Nadjustdist=0.5}
  \node(0)(11,0){$P$}
  \nodelabel[ExtNL=y,NLdist=1,NLangle=-90](0){$\tsm$}
  \node(2)(16,0){2}
  \nodelabel[ExtNL=y,NLdist=1,NLangle=-90](2){3}
  \node(4)(24,0){3}
  \nodelabel[ExtNL=y,NLdist=1,NLangle=-90](4){2}
  \node[Nframe=y,Nw=3.5,Nh=3.5](5)(32,0){4}
  \nodelabel[ExtNL=y,NLdist=1,NLangle=-90](5){0}
  \node(6)(40,0){5}
  \nodelabel[ExtNL=y,NLdist=1,NLangle=-90](6){0}
  \node(7)(48,0){6}
  \nodelabel[ExtNL=y,NLdist=1,NLangle=-90](7){3}
  \drawedge[linegray=1,ELside=r,ELdist=3](2,4){$\scriptstyle\acc$}
  \drawedge[linegray=1,ELside=r,ELdist=3](4,5){$\scriptstyle\acc$}
  \drawedge[linegray=1,ELside=r,ELdist=3](5,6){$\scriptstyle\acc$}
  \drawedge[linegray=1,ELside=r,ELdist=3](6,7){$\scriptstyle\acc$}
\end{gpicture}

\begin{gpicture}[name=tau3,ignore]
  \gasset{Nframe=n,Nw=2,Nh=2,Nfill=y}
  \node[Nfill=n](0)(-5,0){}
  \nodelabel[ExtNL=y,NLdist=1,NLangle=-90](0){$\chi$}
  \nodelabel[ExtNL=y,NLdist=5,NLangle=-90](0){$\ts$}
  \node(1)(00,0){}
  \nodelabel[ExtNL=y,NLdist=1,NLangle=-90](1){1}
  \nodelabel[ExtNL=y,NLdist=5,NLangle=-90](1){0}
  \node(2)(08,0){}
  \nodelabel[ExtNL=y,NLdist=1,NLangle=-90](2){2}
  \nodelabel[ExtNL=y,NLdist=5,NLangle=-90](2){3}
  \node(3)(16,0){}
  \nodelabel[ExtNL=y,NLdist=1,NLangle=-90](3){3}
  \nodelabel[ExtNL=y,NLdist=5,NLangle=-90](3){5}
  \node(4)(24,0){}
  \nodelabel[ExtNL=y,NLdist=1,NLangle=-90](4){4}
  \nodelabel[ExtNL=y,NLdist=5,NLangle=-90](4){6}
  \node(5)(32,0){}
  \nodelabel[ExtNL=y,NLdist=1,NLangle=-90](5){}
  \nodelabel[ExtNL=y,NLdist=5,NLangle=-90](5){8}
  \node(6)(40,0){}
  \nodelabel[ExtNL=y,NLdist=1,NLangle=-90](6){}
  \nodelabel[ExtNL=y,NLdist=5,NLangle=-90](6){8}
  \node(7)(48,0){}
  \nodelabel[ExtNL=y,NLdist=1,NLangle=-90](7){6}
  \nodelabel[ExtNL=y,NLdist=5,NLangle=-90](7){11}
  \drawedge[dash={1}0](1,2){}\drawedge[dash={1}0](2,3){}
  \drawedge(3,4){}\drawedge(4,5){}\drawedge(5,6){}\drawedge(6,7){}
  \drawedge[curvedepth=4,ELpos=45,ELdist=0](1,4){\tiny$[2,\infty]$}
  \drawedge[curvedepth=8,ELpos=50,ELdist=0](1,5){\tiny$[3,\infty]$}
  \drawedge[curvedepth=8,ELpos=50,ELdist=0](2,7){\tiny$[3,\infty]$}
  \drawedge[curvedepth=3,ELpos=40,ELdist=0](3,5){\tiny$[1,3]$}
  \drawedge[curvedepth=3,ELpos=50,ELdist=0](4,6){\tiny$[0,2]$}
  \drawedge[curvedepth=3,ELpos=45,ELdist=0](5,7){\tiny$[1,3]$}
\end{gpicture}

\begin{gpicture}[name=state3,ignore]
  \gasset{Nframe=n,Nw=2.5,Nh=3.5,Nadjustdist=0.5}
  \node(0)(-5,0){$P$}
  \nodelabel[ExtNL=y,NLdist=1,NLangle=-90](0){$\tsm$}
  \node(1)(00,0){1}
  \nodelabel[ExtNL=y,NLdist=1,NLangle=-90](1){0}
  \node(2)(08,0){2}
  \nodelabel[ExtNL=y,NLdist=1,NLangle=-90](2){3}
  \node[Nframe=y,Nw=3.5,Nh=3.5](3)(16,0){3}
  \nodelabel[ExtNL=y,NLdist=1,NLangle=-90](3){1}
  \node(4)(24,0){4}
  \nodelabel[ExtNL=y,NLdist=1,NLangle=-90](4){2}
  \node(7)(48,0){6}
  \nodelabel[ExtNL=y,NLdist=1,NLangle=-90](7){3}
%
%
  \drawedge[linegray=1,ELside=r,ELdist=3](1,2){$\scriptstyle\acc$}
  \drawedge[linegray=1,ELside=r,ELdist=3](2,3){$\scriptstyle\acc$}
  \drawedge[linegray=1,ELside=r,ELdist=3](3,4){$\scriptstyle\acc$}
  \drawedge[linegray=1,ELside=r,ELdist=3](4,7){$\scriptstyle\neg\acc$}
\end{gpicture}

\section{Tree automata for Validity}
\label{sec:AKMV}

In this section, we give one of the most challenging constructions (Theorem
\ref{thm:AKMV}) of the paper, namely, the tree automaton that accepts all valid
and realizable \kTTs which are ``good''.  Good \kTTs are defined below.
In this section, we restrict ourselves to 
closed intervals; that is, those of the form 
$[a,b]$ and $[a, \infty)$, where $a, b \in \mathbb{N}$. 
Fix $K\geq2$.
Not all graphs defined by \kTTs are realizable \TCWs.  Indeed, if $\tau$
is such a \TT, the edge relation $\procrel$ may have cycles or may be
branching, which is not possible in a \TCW. Also, the timing constraints given
by $\matchrel^I$ need not comply with the $\procrel$ relation: for instance, we
may have a timing constraint $e\matchrel^I f$ with $f\procrel^+ e$ 
($\procrel^+$ is the transitive closure of $\procrel$, i.e., $e$ can be reached 
from $f$ after taking $\geq 1$ successor edges $\procrel$).  
Moreover, some terms may define graphs denoting \TCWs which are not realizable.
So we use \AKMV to check for validity.
Since we have only closed intervals in timing
constraints, integer timestamps  suffice for realizability, as 
can be seen from the following lemma 
(Appendix \ref{app:lem-integer}).
\begin{lemma}\label{lem:integer}
  Let $\tcw=(V,\procrel,\lambda,(\matchrel^I)_{I\in\mathcal{I}(M)})$ be a \TCW
  using only closed intervals in its timing constraints.  Then, $\tcw$
  is realizable iff there exists an integer valued timestamp map satisfying all
  timing constraints.
\end{lemma}

Consider a set of colors $P\subseteq\{1,\ldots,K\}$.
For each $i\in P$ we let $i^+=\min\{j\in P\cup\{\infty\}\mid i<j\}$
and $i^-=\max\{j\in P\cup\{0\}\mid j<i\}$. 
If $P$ is not clear from the context, then we write $\nxt_P(i)$ and $\prv_P(i)$.
Given a \kTT $\tau$ with semantics $\sem{\tau}=(G,\chi)$, we denote by
$\Act=\dom(\chi)$ the set of active colors in $\tau$, we let 
$\Right=\max(\Act)$ and $\Left=\min\{i\in\Act\mid\chi(i)\procrel^*\chi(\Right)\}$.
If $\tau$ is not clear from the context, then we write $\Act_\tau$, 
$\Left_\tau$ and $\Right_\tau$.
A \kTT $\tau$ is \emph{good} if 
\begin{itemize}
\item $\tau ::= (a,i)\procrel(b,j) \mid (a,i)\matchrel^{I}(b,j) 
\mid \forget{i}\tau \mid \rename{i}{j}\tau \mid \tau\oplus\tau$,  
  \item for every subterm of the form
$(a,i)\procrel(b,j)$ or $(a,i)\matchrel^{I}(b,j)$ we have $i<j$,
  \item $\rename{i}{j}\tau$ is possible only if  $i^{-}<j<i^{+}$,
  \item $\tau_1\oplus\tau_2$ is allowed if $\Right_1=\Left_2$ and
  $\{i\in\Act_2\mid\Left_1\leq i\leq \Right_1\}\subseteq\Act_1$.
\end{itemize}

\begin{table}[t]
  \noindent\hspace{-3mm}
  \begin{tabular}{|c|c|c|}
    \hline
    $\tau_1$ & $\tau_2$ & $\tau_3$
    \\ \hline
    \Tree[.$\add{1}{5}{\matchrel[3,\infty]}$ [.$\add{3}{5}{\matchrel[1,3]}$ [.$\oplus$ 
      [.$\add{1}{4}{\matchrel[2,\infty]}$ [.$3\procrel 4$ ] ] [.$4\procrel 5$ ] ] ] ]
    \rule{3mm}{0pt}
    &
    \Tree[.$\add{2}{6}{\matchrel[3,\infty]}$ [.$\add{4}{6}{\matchrel[1,3]}$ [.$\oplus$ 
      [.$\add{3}{5}{\matchrel[0,2]}$ [.$4\procrel 5$ ] ] [.$5\procrel 6$ ] ] ] ]
    \rule{5mm}{0pt}
    &
    \Tree[.$\forget{5}$ [.$\oplus$
      [.$\tau_1$ ]
      [.$\rename{3}{4}$ [.$\rename{4}{5}$ [.$\forget{5}$ [.$\tau_2$ ] ] ] ]
    ] ]
    \rule{5mm}{0pt}
    \\ \hline
    ~\gusepicture{tau1} 
    &
    ~\gusepicture{tau2} 
    &
    ~\gusepicture{tau3} 
    \\ \hline 
    \rule{0pt}{8mm}%
    \gusepicture{state1} 
    &
    \gusepicture{state2} 
    &
    \gusepicture{state3} 
    \\ \hline  
  \end{tabular}
  \caption{The second row gives tree representations of three good
  $6$-\TTs $\tau_1$, $\tau_2$, $\tau_3$.  In all these terms, we ignore
  vertex labels and we use $\add{i}{j}{\matchrel I}\tau$ as a macro for
  $\tau\oplus i\matchrel^{I}j$.  The third row gives their semantics
  $\sem{\tau}=(G_\tau,\chi_\tau)$ together with a realization $\ts$, the fourth
  row gives possible states $q$ of \AKMV with $M=4$ after reading the terms.  Here, $L$ is
  the circled color.  The boolean value $\acc(i)$ for each non maximal color $i$
  is written between $\tsm(i)$ and $\tsm(i^{+})$.}
  \protect\label{table:terms-semantics-states}
\end{table}
The intuition is that these good tree terms will give rise to split TC words preserving invariants (I1)-(I3)
of the previous section, following the strategy driven approach.  
 In addition, we ensure that the natural order on the colors is consistent with the linear ordering of the points of the TCW. 

Examples of good \TTs and their semantics are given in
Table~\ref{table:terms-semantics-states}.  Note that the semantics of a \kTT
$\tau$ is a colored graph $\sem{\tau}=(G_\tau,\chi_\tau)$.
Below, we provide a direct construction of a tree automaton, which gives a clear
upper bound on the size of $\AKMV$, since obtaining this bound gets very
technical if we stick to MSO.

\begin{table}[!h]
  \noindent\hspace{-6mm}
  \begin{tabular}{|c|p{120mm}|}
    \hline
    $(a,i)\procrel(b,j)$ 
    &
    $\bot \xrightarrow{(a,i)\procrel(b,j)}q=(P,L,\tsm,\acc)$ 
    is a transition if $i<j$ and $P=\{i,j\}$, $L=i$ and $\acc(j)=\false$.  
    The values for $\tsm(i)$, $\tsm(j)$ and $\acc(i)$ are guessed.
    \\ \hline
    $(a,i)\matchrel^{I}(b,j)$ 
    &
    $ \bot \xrightarrow{(a,i)\matchrel^{I}(b,j)}q=(P,L,\tsm,\acc)$ 
    is a transition if $i<j$ and $P=\{i,j\}$, $L=j$ and $\acc(j)=\false$.  Here, $i$ and $j$ are 
    trivial blocks. 
    The values for $\tsm(i)$, $\tsm(j)$ and $\acc(i)$ are guessed such that
    ($\acc(i)=\true$ and $d(i,j)\in I$) or ($\acc(i)=\false$ and $I.up=\infty$).
    \\ \hline
    $\rename{i}{j}$
    &
    $q=(P,L,\tsm,\acc)\xrightarrow{\rename{i}{j}}q'=(P',L',\tsm',\acc')$ 
    is a transition if $i\in P$ and $i^-<j<i^+$.  
    Then, $q'$ is obtained from $q$ by replacing $i$ by $j$.  
    \\ \hline
    $\forget{i}$
    &
    $q=(P,L,\tsm,\acc)\xrightarrow{\forget{i}}q'=(P',L',\tsm',\acc')$ 
    is a transition if $L<i<\max(P)$ (endpoints should stay colored).  Then,
    state $q'$ is deterministically given by $P'=P\setminus\{i\}$, $L'=L$,
    $\tsm'=\tsm_{\mid P'}$ and $\acc'(i^-)=\ACC(i^-,i^+)\wedge(D(i^-,i^+)<M)$,
    the other values of $\acc'$ are inherited from $\acc$.
    \\ \hline
    $\oplus$
    &
    $q_1,q_2\xrightarrow{\oplus}q$
    where $q_1=(P_1,L_1,\tsm_1,\acc_1)$,
    $q_2=(P_2,L_2,\tsm_2,\acc_2)$ and $q=(P,L,\tsm,\acc)$
    is a transition if
    the following hold 
    \begin{itemize}
      \item $R_1=\max(P_1)=L_2$ and $\{i\in P_2\mid L_1\leq i\leq R_1\}\subseteq
      P_1$ (we cannot insert a new point from the second argument in the
      non-trivial block of the first argument).
      \item $P=P_1\cup P_2$, $L=L_1$, and $\tsm_{\mid P_1}=\tsm_1$ and
      $\tsm_{\mid P_2}=\tsm_2$: these updates are deterministic. In particular, 
      this implies that $\tsm_1$ and $\tsm_2$ coincide on $P_1\cap P_2$.
      \item Finally, $\acc$ satisfies $\acc(\max(P))=\false$ and
      \newline
      $\forall i\in P_1\setminus\{\max(P_1)\}\qquad \acc_1(i)\Longleftrightarrow
      \ACC_q(i,\nxt_{P_1}(i)) \wedge D_q(i,\nxt_{P_1}(i)) < M$
      \newline
      $\forall i\in P_2\setminus\{\max(P_2)\}\qquad \acc_2(i)\Longleftrightarrow 
      \ACC_q(i,\nxt_{P_2}(i)) \wedge D_q(i,\nxt_{P_2}(i)) < M$.
      \newline
      Notice that these conditions imply 
      \newline
      For all $i\in P_1$, if $\nxt_{P}(i)=\nxt_{P_1}(i)$ (e.g., if $L_1\leq
      i<R_1$) then $\acc(i)=\acc_1(i)$.
      \newline
      For all $i\in P_2$, if $\nxt_{P}(i)=\nxt_{P_2}(i)$ (e.g., if $L_2\leq i$)
      then $\acc(i)=\acc_2(i)$. 
       \end{itemize}
    \\ \hline
  \end{tabular}
  
  \caption{Transitions of \AKMV. See Table~\ref{table:terms-semantics-states} 
  and Figure~\ref{fig:small-big} for some intuitions. $I.up$ in row 2 represents upper bound of interval $I$.
  }
  \protect\label{tab:AKMV}
\end{table}

\begin{theorem}\label{thm:AKMV}
  We can build a tree automaton \AKMV with  $M^{\cO(K)}$ number of states such that
  $\Lang(\AKMV)$ is the set of \emph{good} \kTTs $\tau$ such that
  $\sem{\tau}$ is a realizable \TCW and the endpoints of $\sem{\tau}$ are the
  only colored points.
\end{theorem}

\begin{proof} 
  The tree automaton \AKMV reads the \TT bottom-up and stores in its state a
  finite abstraction of the associated graph.  The finite abstraction will keep
  only the colored points of the graph.  We will only accept \emph{good} terms
  for which the natural order on the active colors coincides with the order of
  the corresponding vertices in the final \TCW. The restriction to good terms
  ensures that the graph defined by the \TT is a split-\TCW.

Moreover, to ensure realizability of the \TCW defined by a term, we will guess
timestamps of vertices modulo $M$.  We also guess while reading a subterm
whether the time elapsed between two consecutive active colors is \emph{big}
($\geq M$) or \emph{small} ($<M$).  We see below that the time elapsed is
\emph{small} iff it can be recovered $\mathsf{accurately}$ with the modulo $M$
abstraction.  Then, the automaton has to check that all these guesses are
coherent and using these values it will check that every timing constraint is
satisfied.

Formally, \emph{states of \AKMV} are tuples of the form $q=(P,L,\tsm,\acc)$,
where $P\subseteq\{1,\ldots,K\}$, $L\in P$,
$\tsm\colon P\to[M]=\{0,\ldots,M-1\}$ and $\acc\colon P\to\Bool$.
$\acc$ is a flag which stands for ``accurate'', and is used to check 
if the time elapse between two points is accurate or not, based on the time stamps. 
 
Intuitively, when reading bottom-up a \kTT $\tau$ with
$\sem{\tau}=(V,\procrel,\lambda,(\matchrel^I)_{I\in\mathcal{I}},\chi)$, the automaton \AKMV
will reach a state $q=(P,L,\tsm,\acc)$ such that 
\begin{enumerate}[label=($\mathsf{A}_{\arabic*}$),ref=$\mathsf{A}_{\arabic*}$]
  \item\label{prop:A1} $P=\Act$ is the set of \emph{active} colors in $\tau$,
  $L=\Left$ and $\max(P)=\Right$. 
  
  \item\label{prop:A2} For all $i\in P$, if $L\leq i<\max(P)$ then
  $\chi(i)\procrel^+ \chi(i^+)$ in $\sem{\tau}$.
  
  \item\label{prop:A3} 
  Let ${\hole}=\{(\chi(i),\chi(i^+))\mid i\in P \wedge i<L\}$. 
  This extra relation serves at ordering the blocks of a split-\TCW.
  Then, $(\sem{\tau},\hole)$ is an \emph{ordered} split-\TCW, i.e.,
  ${<}=({\procrel}\cup{\hole})^+$ is a total order on $V$, timing constraints in
  $\sem{\tau}$ are $<$-compatible ${\matchrel}^I\subseteq{<}$ for all $I$, the
  \emph{direct successor} relation of $<$ is ${\lessdot}={\procrel}\cup{\hole}$
  and ${\procrel}\cap{\hole}=\emptyset$.
  Moreover, targets of timing constraints are in the last block: for all
  $u\matchrel^{I}v$ in $(\sem{\tau},\hole)$, we have $\chi(L)\procrel^* v$.

  \item\label{prop:A4} 
  There exists a timestamp map $\ts\colon V\to\N$ such that
  \begin{itemize}
    \item all constraints are satisfied: $\ts(v)-\ts(u)\in I$ for all
    $u\matchrel^Iv$ in $\sem{\tau}$,

    \item time is non-decreasing: $\ts(u)\leq\ts(v)$ for all $u\leq v$,

    \item $(\tsm,\acc)$ is the modulo $M$ abstraction of $\ts$:  $\forall i\in P$ we
    have $\tsm(i)=\ts(\chi(i))[M]$ and $\acc(i)=\true$ iff $i^+\neq\infty$ and
    $\ts(\chi(i^+))-\ts(\chi(i))<M$. 
  \end{itemize}
\end{enumerate}
We say that the state $q$ is a \emph{realizable abstraction} of a term $\tau$ if it satisfies conditions
(\ref{prop:A1}--\ref{prop:A4}).

Indeed, the finite state automaton \AKMV cannot store the timestamp map $\ts$ 
witnessing realizability. Instead, it stores the modulo $M$ abstraction $(\tsm,\acc)$.  We will see that
\AKMV can check realizability based on the abstraction $(\tsm,\acc)$ of $\ts$ and
can maintain this abstraction while reading the term bottom-up.

We introduce some notations.
Let $q=(P,L,\tsm,\acc)$ be a state and let $i,j\in P$ with $i\leq j$.
We define $d(i,j)=(\tsm(j)-\tsm(i))[M]$ and 
$D(i,j)=\sum_{k\in P\mid i\leq k<j}d(k,k^+)$. 
We also define $\ACC(i,j)=\bigwedge_{k\in P\mid i\leq k<j}\acc(k)$.
If the state is not clear from the context, then we write $d_q(i,j)$,
$D_q(i,j)$, $\ACC_q(i,j)$.
For instance, with the state $q_3$ corresponding to the term $\tau_3$ of
Table~\ref{table:terms-semantics-states}, we have $\ACC(1,4)=\true$, $d(1,4)=2$
and $D(1,4)=6=\ts(4)-\ts(1)$ is the $\mathsf{accurate}$ value of the time elapsed.
Whereas, $\ACC(3,6)=\false$ and $d(3,6)=2=D(3,6)$ are both \emph{strict
modulo-$M$ under-approximations} of the time elapsed $\ts(6)-\ts(3)=6$.
The \emph{transitions of \AKMV} are defined in Table~\ref{tab:AKMV}. 

\mypara{Accepting condition} %
The accepting states of \AKMV should correspond to abstractions of \TCWs.  Hence
the accepting states are of the form 
$(\{i,j\},L,\tsm,\acc)$ with $i,j\in\{1,\ldots,K\}$, $i<j$, $L=i$ and
$\acc(j)=\false$.
The correctness of this construction is in Appendix \ref{app:akmv-main}, and 
is obtained by proving (i) the transitions of $\AKMV$ indeed preserve the conditions 
(\ref{prop:A1}--\ref{prop:A4}), (ii) (\ref{prop:A1}--\ref{prop:A4}) ensure among
other things, that the boolean values $\acc(i)$, $\ACC(i,j)$ for $i<j$
indeed defines when the elapse of time is accurately captured by the modulo $M$
abstraction: that is, $\ACC(i,j)$ is true iff the actual time elapse between $i$
and $j$ is captured using the modulo $M$ abstraction $D(i,j)$. 
\end{proof}

\section{Tree automata for timed systems}
\label{sec:TA-sys}

The goal of this section is to build a tree automaton which accepts tree terms 
denoting \TCWs accepted by a TPDA. The existence of a tree automaton can be proved by showing the MSO definability of the 
runs  of the TPDA \Sys on a \TCW. However, as seen in section \ref{sec:AKMV}, we directly construct a tree automaton 
for better complexity. Given the timed system \Sys, let $K$ be the bound on tree-width given by
Theorem~\ref{thm:sw-tpda} and let $M$ be one more than the maximal constant occurring
in the guards of \Sys.
The automaton \AKMS will accept \emph{good} \kTTs with the additional 
restriction that a timing constraint is immediately combined with an existing 
term. That is, \emph{restricted} \kTTs are \emph{good} \kTTs restricted to the 
following syntax: \\
$
\tau ::= (a,i)\procrel(b,j) \mid \stt\oplus [(a,i)\matchrel^{I}(b,j)] 
\mid \forget{i}\tau \mid \rename{i}{j}\tau \mid \tau\oplus\tau
$

\begin{restatable}{theorem}{thmAKMS}\label{thm:AKMS}
  Let \Sys be a TPDA of size $|\Sys|$ (constants encoded in
  unary)
  with set of clocks $X$ and using constants less than $M$.  Let $K$ be the 
  bound on tree-width given by Theorem~\ref{thm:sw-tpda}.
  Then, we can build a tree automaton \AKMS with $|\Sys|^{\cO(K)}\cdot
  K^{\cO(|X|)}$ states such that
  \AKMS accepts the set of \emph{restricted} \kTTs
  $\stt$ such that $\sem{\tau}\in\TCW(\Sys)$. 
  Further, $\TCW(\Sys)=\sem{\Lang(\AKMS)}=\{\sem{\tau}\mid\tau\in\Lang(\AKMS)\}$.
\end{restatable}

\begin{proof}[Proof (Sketch)]
  A state of \AKMS is a tuple $q=(P,L,\delta,\pushb,\popb,G,Z)$ where,
\begin{itemize}
  \item $P$ is the set of active colors, and $L=\mathsf{Left} \in P$ is the left-most point
  that is connected to the right-end-point $R=\mathsf{Right}=\max(P)$ by successor edges on the non-trivial block.

  \item $\delta$ is a map that assigns to each color $k\in P$ the transition $\delta(k)$ guessed at the leaf corresponding to color $k$,

  \item $\pushb$ and $\popb$ are two boolean variables: $\pushb=1$ iff a push-pop edge has been added to $L$ and $\popb=1$ iff a push-pop edge has been added to $R$,
  
  \item $G=(G_x)_{x\in X}$ is a boolean vector of size $|X|$: for each clock $x\in X$, $G_x=1$ iff some constraint on $x$ has already been checked at $R$,

  \item $Z=(Z_x)_{x\in X}$ assigns to each clock $x$ either the color $i\in P$
  with $i<L$ of the unique point on the left of the non-trivial block which is
  the source of a timing constraint $i\matchrel^{I}j$ for clock $x$, or $\bot$
  if no such points exist.
\end{itemize}
For $j\in P$, let $\reset(j)$ be the set of clocks that are reset in the transition $\delta(j)$.
We describe here the most involved kind of transition $q' \oplus q''$ for states $q',q''$. 
The remaining transitions as well as the full proof can be seen in Appendix \ref{app:sys-trans}. 
Let  $q'=(P',L',\delta',\pushb',\popb',G',Z')$,
$q''=(P'',L'',\delta'',\pushb'',\popb'',G'',Z'')$
and $q=(P,L,\delta,\pushb,\popb,G,Z)$. Then 
$q',q''\xrightarrow{\oplus}q$ is a transition if
the following hold:

\centerline{\includegraphics[scale=0.3,page=27]{all-figs}}
\begin{enumerate}[nosep,label= $\mathsf{C}_{\arabic*}$:,ref=$\mathsf{C}_{\arabic*}$]
  \item $R'=\max(P')=L''$ and $\{i\in P''\mid L'\leq i\leq R'\}\subseteq
  P'$ (we cannot insert a new point from the second argument in the
  non-trivial block of the first argument).
  Note that according to $\mathsf{C_1}$, the points
  \protect\includegraphics[scale=0.25,page=28]{all-figs},
  \protect\includegraphics[scale=0.25,page=29]{all-figs} and
  \protect\includegraphics[scale=0.25,page=35]{all-figs} in $P''$ lying between
  $L', R'$ are already points in the non-trivial block connecting $L'$ to $R'$.

  \item $\forall i\in P'\cap P''$, $\delta'(i)=\delta''(i)$ (the guessed
  transitions match).
  By $\mathsf{C}_2$, the transitions $\delta', \delta''$ of
  \protect\includegraphics[scale=0.25,page=28]{all-figs},
  \protect\includegraphics[scale=0.25,page=29]{all-figs} and
  \protect\includegraphics[scale=0.25,page=35]{all-figs} must match.

  \item if there is a $\pushb$ operation in $\delta''(L'')$ then
  $\pushb''=1$ and if there is a pop operation in $\delta'(R')$ then
  $\popb'=1$ (the push-pop edges corresponding to the merging point have
  been added, if they exist). By $\mathsf{C}_3$, if $\delta(R')=\delta(L'')$ contains a pop (resp.\ push)
  operation then $R'=L''$ is the target (resp.\ source) of a push-pop edge.

      \item if some guard $x\in I$ is in $\delta(R')$, then $G'_x=1$ 
      (before we merge, we ensure that the clock guard for $x$ in the transition guessed at $R'$, if any, has been checked). After the merge, $R'=L''$ becomes an internal point; hence by $\mathsf{C}_4$,
any guard $x \in I$ in $\delta'(R')$ must be checked already, i.e., $G'_x=1$.
After the merge, it is no more possible to add an edge $\curvearrowright^I$
leading into $R'$.

      \item if $Z'_x\neq\bot$, then $\forall j\in P''$, $Z'_x<j< L'$ implies
      $x\not\in Reset''(j)$ (If a matching edge starting at $Z'_x<L'$ had been
      seen earlier in run leading to $q'$, then $x$ should not have been reset
      in $q''$ between $Z'_x$ and $L'$, else it would violate the consistency of clocks).  By
      $\mathsf{C}_5$, if $Z'_x$ is
\protect\includegraphics[scale=0.25,page=30]{all-figs} (resp.\
\protect\includegraphics[scale=0.25,page=31]{all-figs}), i.e.,
\protect\includegraphics[scale=0.25,page=30]{all-figs} (resp.\
\protect\includegraphics[scale=0.25,page=31]{all-figs}) is the source of a
timing constraint $\curvearrowright^I$ for clock $x$ whose target is in the
$L'$--$R'$ block, 
then clock $x$ cannot be reset at
\protect\includegraphics[scale=0.25,page=32]{all-figs} and 
\protect\includegraphics[scale=0.25,page=33]{all-figs}
(resp.\ \protect\includegraphics[scale=0.25,page=33]{all-figs}).

      \item if $Z''_x\neq\bot$, then $\forall j\in P'$, $Z''_x<j<L''$ implies
      $x\not\in Reset'(j)$ (If a matching edge starting at $Z''_x<L''$ had been
      seen earlier in run leading to $q''$, then $x$ should not have been reset
      in $q'$ between $Z''_x$ and $L''$).  By $\mathsf{C}_6$, if $Z''_x$ is
\protect\includegraphics[scale=0.25,page=32]{all-figs}, then $x$ cannot be reset
at \protect\includegraphics[scale=0.25,page=31]{all-figs}, 
\protect\includegraphics[scale=0.25,page=36]{all-figs},
\protect\includegraphics[scale=0.25,page=28]{all-figs}, or
\protect\includegraphics[scale=0.25,page=29]{all-figs}.
Likewise, if $Z''_x$ was \protect\includegraphics[scale=0.25,page=33]{all-figs},
then clock $x$ cannot be reset at
\protect\includegraphics[scale=0.25,page=36]{all-figs},
\protect\includegraphics[scale=0.25,page=28]{all-figs}, or
\protect\includegraphics[scale=0.25,page=29]{all-figs}.

      \item $P=P'\cup P''$, $L=L'$, $\delta=\delta'\cup \delta''$,
      $\pushb=\pushb'$, $\popb=\popb''$, $G=G''$ and \newline
      for all $x\in X$ we have $Z_x=Z''_x$ if $Z''_x<L'$, else $Z_x=Z'_x$. 
      $\mathsf{C}_7$ says that on merging, we obtain
the third split-\TCW.
After the merge, if $Z_x$ is defined, it must be on the left of $L'$, i.e., one of
\protect\includegraphics[scale=0.25,page=30]{all-figs},
\protect\includegraphics[scale=0.25,page=32]{all-figs},
\protect\includegraphics[scale=0.25,page=31]{all-figs},
\protect\includegraphics[scale=0.25,page=33]{all-figs}.

      Notice that the above three conditions ensure the well-nestedness of clocks. 
  By $\mathsf{C}_5$ and $\mathsf{C}_6$ we cannot have both
$Z'_x\in\{\protect\includegraphics[scale=0.25,page=30]{all-figs},
\protect\includegraphics[scale=0.25,page=31]{all-figs}\}$ and
$Z''_x\in\{\protect\includegraphics[scale=0.25,page=32]{all-figs},
\protect\includegraphics[scale=0.25,page=33]{all-figs}\}$.  So if
$Z''_x\in\{\protect\includegraphics[scale=0.25,page=32]{all-figs},
\protect\includegraphics[scale=0.25,page=33]{all-figs}\}$ then $Z_x=Z''_x$ and
otherwise $Z_x=Z'_x$ (including when
$Z''_x\in\{\protect\includegraphics[scale=0.25,page=28]{all-figs}, 
\protect\includegraphics[scale=0.25,page=29]{all-figs}\}$ and $Z'_x=\bot$).
\end{enumerate}

\mypara{Accepting Condition} A state $q=(P,L,\delta,\pushb,\popb,G,Z)$ is
accepting if $L=\min(P)$, $\delta(L)$ is some dummy $\varepsilon$-transition 
resetting all clocks and leading to the initial state,
$\mathsf{target}(\delta(R))$ is a final state
and if $\delta(R)$ has a pop operation then $\popb=1$, if it has a
constraint/guard for clock $x$, then $G_x=1$.  Note that the above automaton
only accepts restricted \kTTs; this is sufficient for emptiness checking since
Eve's winning strategy in Section~\ref{sec:stw-tpda} captures all behaviours of
the $\TCW(\Sys)$ while generating only restricted \kTTs.  As a corollary we
obtain (see Appendix~\ref{app:sys-corr}),

\begin{theorem}\label{thm:sys-cor}
  Let \Sys be a TPDA. We have $L(\Sys)\neq\emptyset$ iff
  $L(\AKMV\cap\AKMS)\neq\emptyset$.
\end{theorem}
If the underlying system is a timed automaton, we can restrict the state space
to storing just the tuple $(P,\delta,G)$ as the other components are not
required and $L$ is always $\min(P)$.

\mypara{Possible Extensions} We now briefly explain how to extend our technique
in the presence of diagonal guards: these are guards of the form $x-y \in I$ or
$x-\mathsf{pop}\in I$ or $\mathsf{pop}-x\in I$ where $x, y$ are clocks, and $I$
is a time interval.  The first is a guard that checks the difference between
two clock values, while the other two check the difference between the value of a
clock and the age of the topmost stack symbol at the time of the pop.  To handle
a constraint of the form $x-y \in I$, it is enough to check the difference
between the guessed time stamps at the last reset points of clocks $y, x$ to be
in $I$.  Likewise, to check $x-\mathsf{pop}\in I$ or $\mathsf{pop}-x\in I$, we
check the difference between the guessed time stamps at the points where the top
symbol was pushed on the stack and the last reset of clock $x$.  
Based on the strategy-guided approach for building the tree automaton, 
note that the
last reset points of $x, y$ will not be forgotten until the automaton decides 
to accept; likewise, the push point will not be forgotten until
the pop transition is encountered.  Given this, our construction of the tree
automaton can be extended with the above checks to handle diagonal guards as
well.

\section{Implementation and a case-study}
\label{sec:expt}
We have implemented the emptiness checking procedure for TPDA using our tree-automata based approach,  and describe some results here.  As discussed earlier, the $\mathsf{EXPTIME}$-completeness of this problem for TPDA in general \footnote{note that the $\mathsf{EXPTIME}$ hardness is via poly-time reductions and hence we can $\mathsf{EXPTIME}$ hard and still in $\mathsf{ETIME}$} suggests that in the worst-case, we cannot really hope to do well. However, for certain interesting subclasses of TPDA, we obtain good performance results.

As a concrete subclass, the complexity significantly improves when there is no
extra clock other than the timing constraints associated with the stack; while
popping a symbol, we simply check the time elapsed since the push.  Note that
this can be used to model systems where timing constraints are well-nested: clock resets correspond to push and checking guards corresponds to checking the
age of the topmost stack symbol.  Thus, this gives a technique for reducing the
number of clocks for a timed system with nested timing constraints. For this subclass, the exact number of states of the tree automaton can be
improved to $ 2 \times {(M \times T)}^2 $, where $M$ is 1 plus the maximum
constant, and $T$ is the number of transitions.  This idea can be extended
further to incorporate clocks whose constraints are well-nested with respect to
the stack.  We can also handle clocks which are reset and checked in consecutive
transitions.

For the general model (one stack + any number of clocks), we can use
optimizations to reduce the number of states of the tree automaton to ${(M
\times T)}^{2|X|+2} \times 2^{2|X|+1}$, where $|X|$ is the number of clocks, $M$
is 1 plus the maximum constant and $T$ is the number of transitions.  To see
this, consider the worst case scenario, where a state of the tree automaton has
$|X|$ hanging points and $|X|$ reset points.  In total there can be $2|X|+2$
active points including the left and right end-points of the non-trivial block.
After a combine operation, we can forget a point $i$ of the new state, if it is
the case that every clock $x$ reset at the transition (guessed) at point $i$ is
also reset at some transition at a point after $i$.  Following this strategy, if
we aggressively forget as many points as we can, we will have at most $|X|$
internal (reset) active points between the left and right end-points of the
non-trivial block.  Thus, we reduce the number of active points from
$3|X|+2$ to $2|X|+2$.

As a proof of concept, we have implemented our approach with these
optimizations.  We will now describe some examples we modelled and their
experimental results.  These experiments were run on a 3.5 GHz i5 PC
with 8GB RAM, with number of cores=4.

\subsection*{A Modeling Example : Maze with Constraints}

\begin{figure}[t]
  \begin{center}
  \includegraphics[scale=0.4]{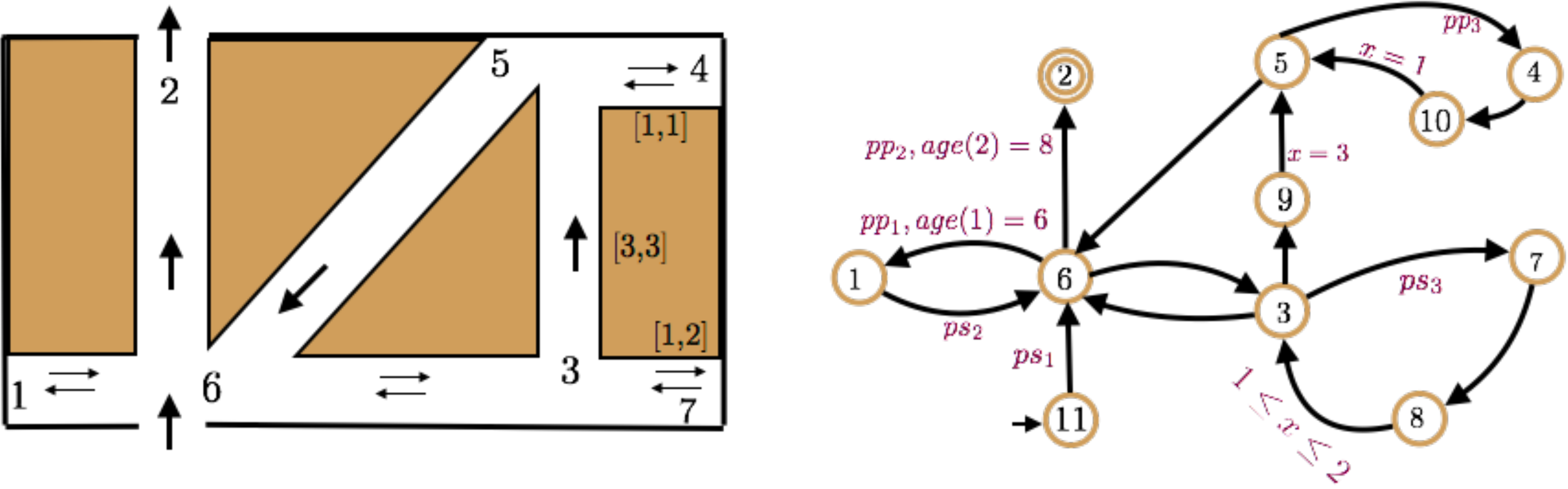}
  \end{center}
  \caption{A simple maze.  Every junction, dead end, entry point or exit point
  is called a place, numbered from 1 to 7.  6 is the entry, 2 the exit, 1, 7 and 4 are dead ends. Time intervals denote the time taken between
  adjacent places; e.g., between 1 and 2 time units must elapse between places 3 and 7. On the right, is the TPDA model of the maze.
  } 
  \protect\label{tpda-maze-tricky-egx}
\end{figure}

As a first interesting example, we model a situation of a robot successfully traversing a maze respecting multiple constraints (see  Figure~\ref{tpda-maze-tricky-egx}).
These constraints may include logical constraints: the robot must visit location
1 before exit, or the robot must load something at a certain place $i$ and
unload it at another place $j$ (so number of visits to $i$ must equal visits to
$j$).  We may also have local and global time constraints which check whether
adjacent places are visited within a time bound, or the total time taken in the
maze is within a given duration.  We show below, via an illustrative example,
that certain classes of such constraints can be converted into a 1-clock TPDA.

One can go from place $p$ to some of its adjacent place $q$ if there is an arrow
from place $p$ to place $q$.  In addition, the following types of constraints
must be respected.
\begin{enumerate}
\item \emph{Logical constraints}
specify certain order between visiting places, the number of times (upper/lower
bounds) to visit a place or places, and so on.  
The logical constraints we have in our example are (a) place 1 must be visited
exactly once, (b) from the time we enter the maze, to visiting place 1, one must
visit place 7 (load) and place 4 (unload) equal number of times, and at any
point of time, the number of visits to place 7 is not less than number of visits
to place 4.  (c) from visiting place 1 to exiting the maze, one must visit place
7 and place 4 equal number of times and, at any point during time, number of
visits to place 7 is not less than number of visits to place 4.
\item \emph{Local time constraints} specify time intervals which must
be respected while going from a place to its adjacent place.  The time taken
from some place $i$ to another adjacent place $j$ is given as a closed interval
$[a,b]$ along with the arrow.  One cannot spend any time between a
pair of adjacent places other than the ones specified in the maze. 
For example, the time bound for going from place 7 to 3 is given, while the time
taken from place 3 to place 7 and place 6 to place 1 is zero ([0,0]), since it
is not mentioned.  Further, one cannot stay in any place for non-zero duration.
  
\item \emph{Global time constraints} specify the total time that can be
elapsed between visiting any two places.  From entering into the maze to
visiting of place 1, time taken should be exactly $m$ units (a parameter).  From
visiting place 1 to exit, time bound should be exactly $n$ units (another
parameter).
\end{enumerate}

A maze respecting multiple constraints as above is converted into a 1-clock
TPDA. While the details of this conversion are given in Appendix \ref{app:exp},
the main idea is to encode local time bounds with the clock which is reset on
all transitions.  A logical constraint specifying equal number of visits to
places $p_1, p_2$ is modelled by pushing symbols while at $p_1$, and popping
them at $p_2$.  Likewise, if there is a global time constraint that requires a
time elapse in $[a,b]$ between the entry and some place $p$, then push on the
stack at entry, and check its age while at $p$.  Note that all these are
\emph{well-nested} properties.  

To check the existence of a legitimate path in the maze respecting the
constraints, our tool checks the existence of a run in the TPDA. By running our
tool on the TPDA constructed (and fixing the parameters to be $m=7$, $n=8$), we
obtain the following run: (described as a sequence of pairs the form :
State, Entry time stamp in the state) 

\noindent %
$(6, 0.0) \rightarrow (3, 0.0)
\rightarrow (7, 0.0) \rightarrow (3, 1.0) \rightarrow (7, 1.0) \rightarrow (3,
2.0) \rightarrow (5, 5.0) \rightarrow (4, 5.0) \rightarrow (5, 6.0) \rightarrow
(4, 6.0) \rightarrow (5, 7.0) \rightarrow (6, 7.0) \rightarrow (1, 7.0)
\rightarrow (6, 7.0) \rightarrow (3, 7.0) \rightarrow (7, 7.0) \rightarrow (3,
9.0) \rightarrow (7, 9.0) \rightarrow (3, 10.0) \rightarrow (5, 13.0)
\rightarrow (4, 13.0) \rightarrow (5, 14.0) \rightarrow (4, 14.0) \rightarrow
(5, 15.0) \rightarrow (6, 15.0) \rightarrow (2, 15.0)$ 

\noindent
\begin{minipage}{.42\textwidth}
  The scalability is assessed by instantiating the maze for various choices of
  maximum constants used, as well as number of transitions.  The running times
  with respect to various choices for the maximum constant are plotted on the 
  right.  More maze examples can be found in Appendix \ref{app:exp}.
\end{minipage}
\hfill
\begin{minipage}{.55\textwidth}
  \includegraphics[scale=0.5]{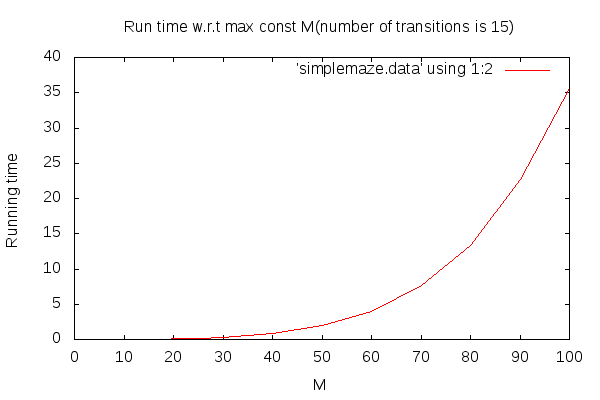} 
\end{minipage}

\section{Conclusion}
\label{sec:disc}
We have obtained a new construction for the emptiness checking of TPDA, using tree-width. 
The earlier approaches  \cite{lics12}, \cite{lata12} which handle dense time and discrete time 
push down systems respectively use an adaptation of the well-known idea of timed regions.  The technique in \cite{lata12} does not extend 
to dense time systems, and it is not clear whether the approach in \cite{lics12} will work for 
say, multi stack push down automata even with bounded scope/phase restrictions. Unlike this, our approach is uniform : all our proofs except the tree automaton for realizability already work even if we have open guards.  
Our realizability proof has to be adapted for open guards and this is work under progress;
 in this paper, we focussed on closed guards to obtain an efficient 
tool based on our theory.  Likewise, our proofs can be extended to  bounded phase/scope/rounds 
multi stack timed push down automata : we need to show a bound on the tree-width, and then adapt 
the tree automaton construction for the system automaton. The tree automaton checking realizability requires no change.   
Further, our $\mathsf{ETIME}$  complexity is the best known upper bound, as far as we know. 
  With the theoretical improvements in this paper, we implement our approach and examine its performance on real examples.  To the best of our knowledge, this is the first tool implementing timed push down systems.   
   We plan to optimize our implementation to get a more robust and scalable tool :  for instance, when the language is non-empty, a witness for non-emptiness can be produced. For the subclasses we have, it would be good to have a characterization and automatic translation (currently this is done by hand) that replaces well-nested clock constraints by stack edges, and thus leading to better implementability. We also plan to extend our implementations to give bounded under-approximations for timed automata with multiple stacks that can be used to model and analyze recursive programs with timers.\\

\noindent{Acknowledgement.} The authors thank Vincent Jug\'e for insightful discussions on the MSO definability of realizability of TC words.   

\bibliographystyle{plain}
\bibliography{papers}
\newpage
\appendix
\centerline{\bf {\Large {Appendix}}}
\section{Tree Algebra : An Example}\label{app:TT-example}
Figure \ref{tt-decomp} shows the construction of a \TCW using tree terms.  The
example uses 4 colors: think of the shapes as colors. The example starts out with 3 atomic terms, and 
builds the \TCW using the operations of combine and forget. In each step, the resultant 
\TCW obtained is drawn in the dialogue box.

 \begin{figure}[h]
\begin{center}
\includegraphics[scale=.25,page=9]{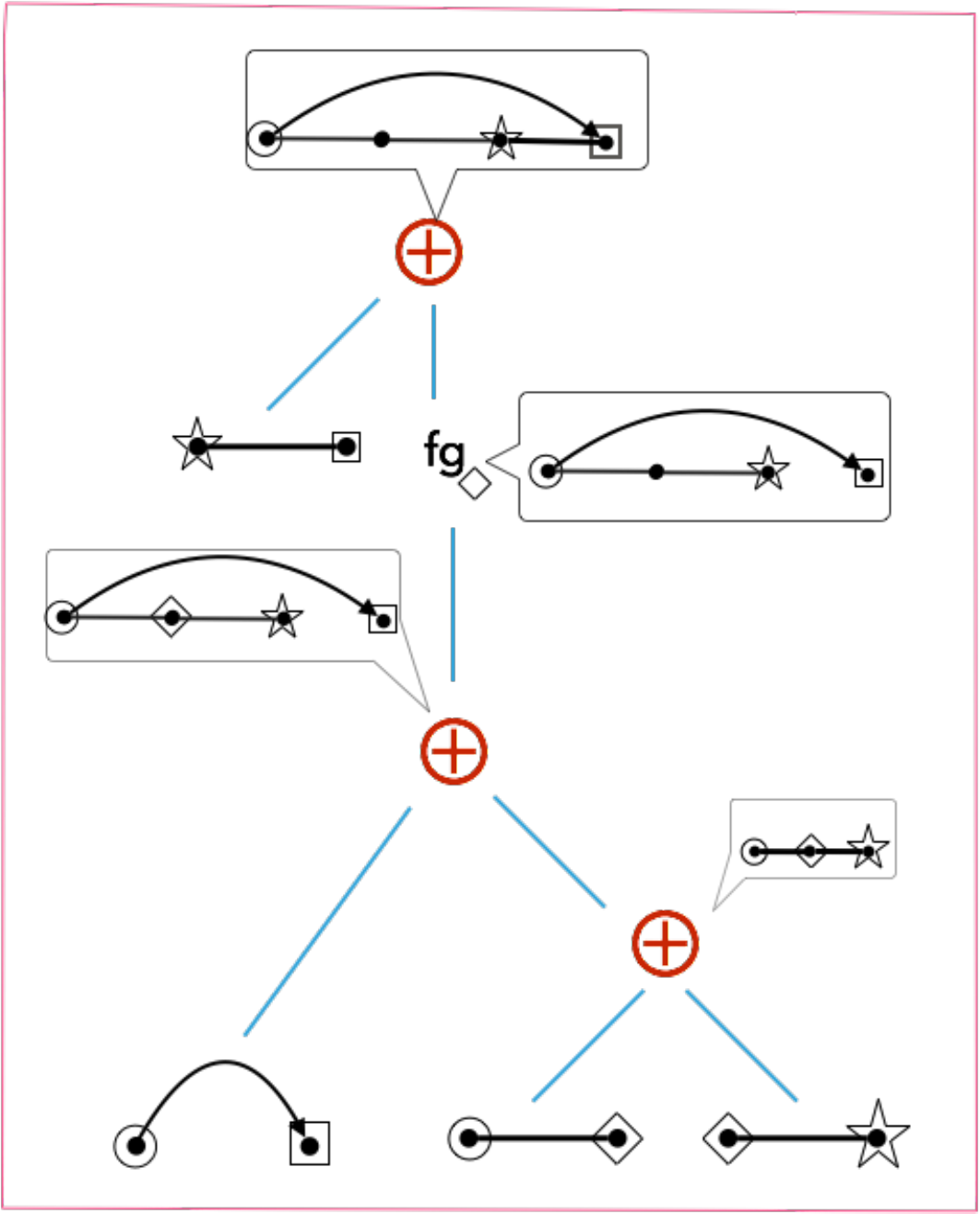}
\end{center}
\caption{An example constructing a \TCW using atomic tree terms}
\label{tt-decomp}
\end{figure}

\section{Proof of Lemma \ref{TPDA-bound}}
\label{app:sw}

We prove this by playing the ``split game'' between \emph{Adam} and \emph{Eve}
in which \emph{Eve} has a strategy to disconnect the word without introducing
more than $3|X|+3$ colors.  \emph{Eve}'s strategy processes the word from right
to left.  Starting from any \TCW, Eve colors the end points of the \TCW, as well
as the last reset points (from the right end) corresponding to each clock.  Here
she uses at most $|X|+2$ colors.  On top of this, depending on the last point,
we have different cases.  A detailed example of the split-game is given in 
Table~\ref{tab:split-game-example}.

\begin{table}
  Consider the following \TCW with two clocks.  The green and blue edges
  represent two clocks while the red edges represent the stack.

  \centerline{
   \includegraphics[scale=.3,page=15]{all-figs}
  }

  Step 1: In the figure below, Eve adds colors to end points, and to last reset
  points.

  \centerline{
   \includegraphics[scale=.3,page=4]{all-figs}
  }

  Step 2: Eve removes the last clock edge, and adds a color to the before last
  point.

  \centerline{
   \includegraphics[scale=.3,page=3]{all-figs}
  }  

  Step 3: This enables a divide, resulting in the removal of
  \includegraphics[scale=.2,page=16]{all-figs}.  The last point of the resultant
  word will be \includegraphics[scale=.2,page=17]{all-figs}.  This point is also
  removed after removing the clock edge, and adding a color to the preceding
  point, making it \includegraphics[scale=.2,page=23]{all-figs}.  This will be
  the last point now, and is the target of a stack edge.
  
  \medskip %
  Step 4: Eve adds a color to the source of the stack edge and to
  the reset points of clock edges which cross over the stack.  This enables a
  divide, resulting in two words shown below

  \centerline{
  \includegraphics[scale=.3,page=22]{all-figs}  
  }

  Step 5: Yet another divide when the last point is the target of a stack edge,
  and the source of the stack edge is an internal point.  In this case, both
  resultant words have trivial blocks and one non-trivial block.

  \vspace{-36mm}
  \centerline{
  \includegraphics[scale=.3,page=24]{all-figs} 
  }

  \caption{Examples for the split game.}
  \protect\label{tab:split-game-example}
\end{table}

\begin{enumerate}
  \item If the last point is the target of a $\curvearrowright^{x}$ relation for
  some clock $x$, then Eve's strategy is to divide the \TCW by removing the
  clock edge (Step~2 in Table~\ref{tab:split-game-example}).  Notice that the
  source and target points of this clock edge are colored.  This results in two
  \TCWs, one of which is atomic, consisting of the matching clock edge, while in
  the other \TCW, the last point is no longer the target of a matching relation
  for $x$.  We apply Case 1 until the last point of the \TCW is no longer the
  target of a clock edge $\matchrel^{z}$.

  \item If the last point is not the target of a $\curvearrowright$ relation,
  then \emph{Eve} adds a color to the immediate previous point and divides the
  \TCW removing the atomic edge consisting of the last two points (Step~3 of
  Table~\ref{tab:split-game-example}).  Note that when this last point happens to
  be a reset point for some clock $x$, then while removing this point as
  explained above, Eve also adds a color to the new last reset point for $x$ in
  the resulting \TCW. This ensures that the last reset points for every clock
  $x$ is colored.

  \item The most interesting case is when the last point is the target of a
  $\curvearrowright^{s}$ relation of the stack.  There are two cases here.  The
  simple case is when the source and target nodes of the stack edge are the end
  points of a non-trivial block in the \TCW. In this case, Eve simply
  disconnects the stack edge.

  The harder case is when the source of the stack edge is an internal point.  In
  this case, unlike the removal of the clock edge, adding a color to the source
  of the stack edge and removing it is not a good strategy since we might have
  unboundedly many stack edges, resulting in the use of unboundedly many colors.
  To keep a bound on the number of colors needed, Eve divides the \TCW as
  follows:
  \begin{itemize}
    \item First Eve adds a color to the  source of the stack edge

    \item If there are any clock edges crossing this stack edge, Eve adds colors
    to the corresponding reset points.  Note that this results in adding atmost
    $|X|$ colors.
    
    \item Eve disconnects the \TCW into two parts, such that the right part
    $\tcw_2$ consists of one non-trivial block whose end points are the source
    and target points of the stack edge, and also contains to the left of this
    block, atmost $|X|$ trivial blocks.  Each of these trivial blocks are the
    reset points of those clock edges which cross over.  The left part $\tcw_1$
    is a \TCW consisting of all points to the left of the source of the stack
    edge, and has all remaining edges other than the clock edges which have
    crossed over.  Adam can now continue the game choosing $\tcw_1$ or $\tcw_2$.
    We illustrate this case below.  Note that in one of the words so obtained,
    the stack edge completely spans the non-trivial block, and can be easily
    removed.
  \end{itemize}
\end{enumerate}

\noindent{\emph{Invariants and bound on tree-width}}.  We now discuss some
invariants on the structure of the split-\TCWs as we play the game using the
above described strategy of Eve.  The last two split \TCWs in
Table~\ref{tab:split-game-example} are representatives of the split \TCWs that
may occur during the game after a divide operation.  These \TCWs satisfy the
following invariants:
\begin{itemize}
\item[(I1)] We have at most $|X|$ colored trivial blocks to the left of the only non-trivial block,  
\item[(I2)] The last reset node of each clock on the non-trivial block is colored
\item[(I3)]The end points of the non-trivial block are colored
\end{itemize}
To maintain the above invariants, we need $|X|+1$ extra colors than the at most
$2|X|+2$ mentioned above.  To prove the bound on the tree-width, we use the
following lemma.

\begin{lemma}
  Any split \TCW formed during the game has exactly one non-trivial block.  It
  is divided using Eve's strategy described above using at most $3|X|+3$ colors,
  and the resultant words satisfy (I1-I3).  The tree-width is hence at most
  $3|X|+2$.
\end{lemma}

\begin{proof}
  Let us start from a split \TCW satisfying (I1)-(I3).  The case when the last
  point is not the target of a clock/stack edge is easy: we simply add one color
  to the predecessor of the last point, and detach the linear edge.  If this
  point was a reset point of say clock $x$, then we need to add a color to the
  last reset point of $x$ to maintain the invariant.  Note that Eve's strategy
  does this.  If the last point is the target of a clock edge, then simply
  removing the clock edge suffices.  In both cases, the invariant is satisfied.

  Now consider the case when the last point is the target of a stack edge.  The
  simple case is when the stack spans the non-trivial block.  In this case, we
  simply remove the stack edge, preserving the invariant.  If the source of the
  stack is an internal point in the non-trivial block, then as in steps 4 and 5 
  of Table~\ref{tab:split-game-example}, the resulting split-\TCWs
  $\tcw_1$ and $\tcw_2$ will consist of a non-trivial block, and
  trivial blocks to its left, corresponding to the resets of clocks whose edges
  cross over to the block.  The number of colors in $\tcw_2$ is hence at most
  $2|X|+2$ ($|X|+2$ on the non-trivial block and $|X|$ to the left of the
  block).  Whenever the right most point is the target of a stack edge which
  does not span the non-trivial block, then Eve has to repeat case 3 above by
  (1) adding a color to the source of the stack edge, (2) color the reset points
  to the left of the source of the stack edge, of those clocks whose edges cross
  over.  Note that this needs introducing at most $|X|+1$ colors.  The split
  \TCW $\tcw$ now uses at most $3|X|+3$ colors.

  Let $S$ ($|S| \leq |X|$) be the set of trivial blocks/hanging reset points to
  the left of the block in $\tcw$.  On division of the (split)-$\TCW$ $\tcw$, we
  obtain the right part as the \TCW $\tcw_2$ containing the last stack edge of
  $\tcw$ spanning the non-trivial block; along with a set $S_2$ of hanging reset
  points to the left of the block.  Some of the points of $S_2$ could be from
  $S$ while the remaining $|S_2\setminus S|$ are the reset points of clocks
  which were on the non-trivial block of $\tcw$ to the left of the stack edge,
  and whose clock edges were crossing the stack edge.
The left part $\tcw_1$ consists of hanging reset points  $S_1$  along with a non-trivial block whose left end point 
is the same as the left end point of $\tcw$, and whose right end point 
is the source of the stack edge in $\tcw$. Note that $S_1 \subseteq S$; 
(1) If  there is a reset point in $S$ all of whose clock edges 
cross the stack edge, then this reset point will be in $S_2 \setminus S_1$  
(2) If  there is a reset point in $S$  all of whose clock edges 
are to the left of the stack edge, then this reset point will be in $S_1 \setminus S_2$, and 
(3) If  there is a reset point in $S$ such that some clock edges 
cross the stack edge, while some dont, it will be in  $S_2  \cap S_1$. 
(I1)-(I3) are satisfied for $\tcw_1, \tcw_2$, and hence the number of colors
in both is at most $|X|+|X|+2$.

In our running example, when we encounter the next stack edge, Eve needs to add
two colors, one (the triangle) for the source of the stack edge, and one (the
pentagon) for the clock which crosses the stack edge.  The last reset point of
the other clock is still the one hanging to the left of the block (colored
diamond).  On division, we obtain two split \TCW $\tcw_1, \tcw_2$ both having
the same form: a sequence of $y \leq |X|$ hanging points to the left of a block.
Each of these hanging points have a clock edge whose target lies in the
non-trivial block.  Whichever of these words is chosen by Adam, the subsequent
split-\TCWs obtained during the game will continue to have the normal form and
the number of colors needed before any divide is at most $3|X|+3$.
\end{proof}

\noindent{\emph{Remark}.} Note that if there are no stack edges in the \TCW,
then Eve's strategy is simply to keep colors on the last reset points (from the
right end) for all clocks and on the right point (no need to color the leftmost
point of the word).  It is necessary to keep colors at the last reset points of
all the clocks in order to divide the \TCW, since any of the clocks can be
checked at the last point.  This results in the use of $|X|+1$ colors.  If the
last point is the target of a matching clock edge, then division is just
removing the clock edge.  If not, then we add one color to the predecessor of
the last point, and remove the last linear edge.  This results in the use of
$|X|+2$ colors.  Note that the division of the \TCW is always by removing an
atomic term in this case, which may be a linear edge or a clock edge, as a
result of which, one of the words obtained after divsion is always atomic.  Adam
will hence always choose the left word to prolong the game.  The word chosen by
Adam is always a single block with no hanging points to the left.  This already
shows the tree-width to be at most $|X|+1$ in case of timed automata.

\section{Proofs from Section \ref{sec:AKMV}}
\label{app:AKMV}
\subsection{Proof of Lemma \ref{lem:integer}}
\label{app:lem-integer}
\begin{proof}
  Consider two non-negative real numbers $a,b\in\mathbb{R}_+$ and let
  $i=\lfloor a\rfloor$ and $j=\lfloor b\rfloor$ be their integral parts. Then, 
  $j-i-1<b-a<j-i+1$. It follows that for all closed intervals $I$ with integer 
  bounds, we have $b-a\in I$ implies $j-i\in I$.
  
  Assume there exists a non-negative real-valued timestamp map $\ts\colon
  V\to\mathbb{R}_+$ satisfying all timing constraints of $\tcw$.  Since all the time constraints are closed,  we deduce that $\lfloor\ts\rfloor\colon V\to\N$ also realizes all
  timing constraints of $\tcw$.
  The converse direction is clear.  
\end{proof}

\subsection{Proof of Theorem \ref{thm:AKMV}}
\label{app:akmv-main}
To prove Theorem \ref{thm:AKMV}, we prove the following.

\begin{claim}\label{claim:real-abstr}
  Let $q$ be a state and $\tau$ be a \TT. Then the following hold:
  \begin{enumerate}
    \item Assume (\ref{prop:A1}--\ref{prop:A3}) are satisfied.  Then, for all
    $i,j\in P$ we have $i<j$ iff $\chi(i)<\chi(j)$: the natural ordering on
    colors coincide with the ordering of colored points in the split-\TCW
    $(\sem{\tau},\hole)$.
  
    \item  Assume that $\ts$ is a timestamp map satisfying items 2 and 3 of \eqref{prop:A4}. 
    Then, for all $i,j\in P$ such that $i\leq j$, we have
    $d(i,j)=D(i,j)[M]=(\ts(\chi(j))-\ts(\chi(i)))[M]$ and
    $d(i,j)\leq D(i,j)\leq\ts(\chi(j))-\ts(\chi(i))$ 
    ($d$ and $D$ give modulo $M$ under-approximations of the actual time
    elapsed).  Moreover, $\ACC$ tells whether $D$ gives the $\mathsf{accurate}$
    elapse of time:
    \begin{align*}
      \acc(i)=\true & \Longleftrightarrow d(i,i^+)=\ts(\chi(i^+))-\ts(\chi(i)) 
      \\
      \ACC(i,j)=\true & \Longleftrightarrow D(i,j)=\ts(\chi(j))-\ts(\chi(i)) \\
      \ACC(i,j)=\false & \implies \ts(\chi(j))-\ts(\chi(i)) \geq M 
    \end{align*}
  \end{enumerate}
\end{claim}
\begin{proof}
  1.  From (\ref{prop:A2}--\ref{prop:A3}) we immediately get $\chi(i)<\chi(i^+)$
  for all $i\in P\setminus\{\max(P)\}$.  By transitivity we obtain
  $\chi(i)<\chi(j)$ for all $i,j\in P$ with $i<j$.  Since $<$ is a strict total
  order on $V$, we deduce that, if $\chi(i)<\chi(j)$ for some $i,j\in P$, then
  $j\leq i$ is not possible.
  
  \smallskip\noindent
  2. Let $i,j\in P$ with $i\leq j$. Using items 2 and 3 of \eqref{prop:A4} we get
  $$
  d(i,j) = (\tsm(j)-\tsm(i))[M] = (\ts(\chi(j))[M]-\ts(\chi(i))[M])[M] = 
  (\ts(\chi(j))-\ts(\chi(i)))[M] \,.
  $$
  Applying this equality for every pair $(k,k^+)$ such that $i\leq k<j$ we get 
  $D(i,j)[M]=(\ts(\chi(j))-\ts(\chi(i)))[M]$. Since $\ts$ is non-decreasing 
  (item 2 of \ref{prop:A4}), it follows that 
  $d(i,j)\leq D(i,j)\leq\ts(\chi(j))-\ts(\chi(i))$.
  
  Now, using again \eqref{prop:A4} we obtain $\acc(k)=\true$ iff
  $d(k,k^+)=\ts(\chi(k^+))-\ts(\chi(k))$. Applying this to all $i\leq k<j$ we 
  get $\ACC(i,j)=\true$ iff $D(i,j)=\ts(\chi(j))-\ts(\chi(i))$.
  
  Finally, $\ACC(i,j)=\false$ implies $\acc(k)=\false$ for some $i\leq k<j$.  Using
  \eqref{prop:A4} we obtain $\ts(\chi(j))-\ts(\chi(i)) \geq
  \ts(\chi(k^+))-\ts(\chi(k)) \geq M$.
\end{proof}

Next we show that the transitions of $\AKMV$ indeed preserve the conditions 
\eqref{prop:A1} to \eqref{prop:A4}. That is, any run of $\AKMV$ is such that \ref{prop:A1}--\ref{prop:A4}
hold good.

\begin{lemma}\label{lem:AKMV-tcw}
  Let $\tau$ be a \kTT and assume that $\AKMV$ has a run on $\tau$
  reaching state $q$.  Then, $\tau$ is \emph{good} and $q$ is a realizable
  abstraction of $\tau$.
\end{lemma}

\begin{proof}
  The conditions on the transitions for $\rename{i}{j}$, $\add{i}{j}{\procrel}$ 
  and $\add{i}{j}{\matchrel I}$ in Table \ref{tab:AKMV}
directly ensure that the term is good.
  We show that (\ref{prop:A1}--\ref{prop:A4}) are maintained by transitions of \AKMV.
\begin{itemize}
  \item Atomic \TTs (1): Consider a transition $\bot \xrightarrow{(a,i)\procrel(b,j)}q$ of \AKMV.
  
  It is clear that $q$ is a realizable abstraction of the term 
  $\tau=(a,i)\oplus(b,j)\oplus i\procrel j$.

  \item Atomic \TTs (2): Consider a transition $\bot \xrightarrow{(a,i)\matchrel^{I}(b,j)}q$ of \AKMV.
  
  It is clear that (\ref{prop:A1}--\ref{prop:A3}) are satisfied for state $q$
  and term $\tau=(a,i)\oplus(b,j)\oplus i\matchrel^{I}j$. 
  
  Now, we define $\ts(\chi(i))=\tsm(i)$ and $\ts(\chi(j))=
  \begin{cases}
    \tsm(i)+d(i,j) & \text{if \acc(i)=\true}  \\
    \tsm(i)+d(i,j)+M & \text{otherwise.}
  \end{cases}$
  
  Using the last condition ($\acc(i)=\true$ and $d(i,j)\in I$) or
  ($\acc(i)=\false$ and $I.up=\infty$) of the transition we can easily check
  that \eqref{prop:A4} is satisfied.
  
  \item $\rename{i}{j}$:
  Consider a transition $q\xrightarrow{\rename{i}{j}}q'$ of \AKMV.
  
  Assume that $q$ is a realizable abstraction of some $K$-\TT $\tau$ and 
  let $\tau'=\rename{i}{j} \tau$. It is easy to check that $q'$ is a 
  realizable abstraction of $\tau'$.

  \item $\forget{i}$:
  Consider a transition $q\xrightarrow{\forget{i}}q'$ of \AKMV.

  Assume that $q$ is a realizable abstraction of some $K$-\TT $\tau$ and 
  let $\tau'=\forget{i} \tau$. It is easy to check that $q'$ is a 
  realizable abstraction of $\tau'$. In particular, the correctness of the 
  update $\acc'(i^-)$ follows from Claim~\ref{claim:real-abstr}.

  \item $\oplus$:
  Consider a transition $q_1,q_2\xrightarrow{\oplus}q$ of \AKMV.

  Assume that $q_1$ and $q_2$ are realizable abstractions of some 
  $K$-\TTs $\tau_1$ and $\tau_2$ with timestamp maps $\ts_1$ and $\ts_2$ 
  respectively. Let $\tau=\tau_1\oplus\tau_2$. We show that $q$ is a 
  realizable abstraction of $\tau$. 

    \medskip\noindent\eqref{prop:A1} %
    We have $\Act_\tau=\Act_{\tau_1}\cup\Act_{\tau_2}=P_1\cup P_2=P$.
    Moreover, using $R_1=\max(P_1)=L_2$, we deduce that $L=L_1=\Left_\tau$ and 
    $\max(P)=\max(P_2)=\Right_\tau$.
        
    \medskip\noindent\eqref{prop:A2} %
    Let $i\in P$ with $L\leq i<\max(P)$. Either $i<\max(P_1)$ and 
    we get $i\in P_1$ and $j=\nxt_P(i)=\nxt_{P_1}(i)$.  We deduce that 
    $\chi(i)=\chi_1(i)\procrel^{+}\chi_1(j)=\chi(j)$. Or $L_2=\max(P_1)\leq i$ and 
    we get $i\in P_2$ and $j=\nxt_P(i)=\nxt_{P_2}(i)$.  We deduce that 
    $\chi(i)=\chi_2(i)\procrel^{+}\chi_2(j)=\chi(j)$.
    
    \medskip\noindent\eqref{prop:A3} %
    Let ${\hole}=\{(\chi(i),\chi(\nxt_P(i)))\mid i\in P \wedge i<L\}$.  Let
    ${<}=({\procrel}\cup{\hole})^+$.  Using \eqref{prop:A2} and the definition
    of $<$, it is easy to see that for all $i,j\in P$, if $i<j$ then
    $\chi(i)<\chi(j)$.  Using ${\procrel}={\procrel_1}\uplus{\procrel_2}$, we
    deduce that ${<_1}\cup{<_2}\subseteq{<}$.
    
    Let $u\matchrel^I v$ be a timing constraint in $\sem{\tau}$.  Either it is
    in $\sem{\tau_1}$ and it is compatible with $<_1$, hence also with $<$.  Or
    it is in $\sem{\tau_2}$ and it is compatible with $<_2$ and with $<$.
    
    Using conditions $R_1=\max(P_1)=L_2$ and $\{i\in P_2\mid L_1\leq i\leq
    R_1\}\subseteq P_1$ of the transition, we deduce easily that $<$ is a total
    order on $V$ and that $(\sem{\tau},\hole)$ is a split-\TCW.
    
    Also, since the last block of $(\sem{\tau_1},\hole_1)$ is concatenated with 
    the last block of $(\sem{\tau_2},\hole_2)$, targets of timing constraints 
    are indeed in the last block of $(\sem{\tau},\hole)$.

    \medskip\noindent\eqref{prop:A4} %
    We construct the timestamp map $\ts$ for $\tau$ inductively on 
    $V=V_1\uplus V_2$ following the successor relation 
    ${\lessdot}={\procrel}\cup{\hole}$. If 
    $v=\min(V)$ is the first point of the split-\TCW, we let 
    $$
    \ts(v)=\begin{cases}
      \ts_1(v) & \text{if } v\in V_1 \\
      \ts_2(v) & \text{otherwise.}
    \end{cases}
    $$
    Next, if $\ts(u)$ is defined and $u\procrel v$ then we let
    $$
    \ts(v)=\begin{cases}
      \ts(u) + \ts_1(v)-\ts_1(u) & \text{if } u,v\in V_1 \\
      \ts(u) + \ts_2(v)-\ts_2(u) & \text{if } u,v\in V_2 \,.
    \end{cases}
    $$
    Finally, if $\ts(u)$ is defined and $u\hole v$ then, with $i,j\in P$ 
    being the colors of $u$ and $v$ ($\chi(i)=u$ and $\chi(j)=v$), we let
    $$
    \ts(v)=\begin{cases}
      \ts(u) + d_q(i,j) & \text{if } \acc(i)=\true \\
      \ts(u) + d_q(i,j)+M & \text{if } \acc(i)=\false \,.
    \end{cases}
    $$
    With this definition, the following hold
    \begin{itemize}
      \item Time is clearly non-decreasing: $\ts(u)\leq\ts(v)$ for all $u\leq v$
      
      \item $(\tsm,\acc)$ is the modulo $M$ abstraction of $\ts$.  The proof is
      by induction. 
      
      First, if $i=\min(P)$ then $v=\min(V)=\chi(i)$ and $i\in
      P_1$ iff $v\in V_1$.  Using the definitions of $\tsm$ and $\ts$, we
      deduce easily that $\tsm(i)=\ts(v)[M]$.  
      
      Next, let $i\in P$ with $j=\nxt_P(i)<\infty$.  Let $u=\chi(i)$,
      $v=\chi(j)$ and assume that $\tsm(i)=\ts(u)[M]$.
      
      If $u\procrel^+v$ and $u,v\in V_1$ then $i,j\in P_1$ and 
      \begin{align*}
        \ts(v)[M] &= (\ts(u)[M] + \ts_1(v)[M]-\ts_1(u)[M])[M] \\
        &= (\tsm(i)+\tsm_1(j)-\tsm_1(i))[M]
        = \tsm(j) \,.
      \end{align*}
      Moreover, $j=\nxt_P(i)=\nxt_{P_1}(i)$ and we get $\acc(i)=\acc_1(i)$. 
      Also, $\ts(v)-\ts(u)=\ts_1(v)-\ts_1(u)$. We 
      deduce that $\acc(i)=\true$ iff $\ts(v)-\ts(u)< M$.
      The proof is similar if $u\procrel^+v$ and $u,v\in V_2$.
      
      Now, if $u\not\procrel^+v$ then $u\hole v$ (endpoints are always 
      colored). We deduce that
      $$
      \ts(v)[M] = (\ts(u)[M] + d_q(i,j))[M] = (\tsm(i)+\tsm(j)-\tsm(i))[M] =
      \tsm(j) \,.
      $$
      Moreover, it is clear that $\ts(v)-\ts(u)< M$ iff $\acc(i)=\true$.
            
      \item Constraints are satisfied. Let $u\matchrel^Iv$ be a timing 
      constraint in $\sem{\tau}$. Wlog we 
      assume that $u,v\in V_1$. We know that $\ts_1(v)-\ts_1(u)\in I$.
      
      If $u\procrel^+v$ then we get $\ts(v)-\ts(u)=\ts_1(v)-\ts_1(u)$ from the
      definition of $\ts$ above.  Hence, $\ts(v)-\ts(u)\in I$.
      
      Now assume there are holes between $u$ and $v$ in $(\sem{\tau},\hole)$.
      Then, we have $u=\chi(i)$ for some $i\in P_1$ with $i<L=L_1$.
      Since targets of timing constraints are always in the last block 
      \eqref{prop:A3}, we get $v'=\chi(L_1)\procrel^* v$.
      
      We deduce from the definition of $\ts$ that
      $\ts(v)-\ts(v')=\ts_1(v)-\ts_1(v')$.  Now, using Claim~\ref{claim:oplus}
      below we obtain:
      \begin{itemize}
        \item Either $\ACC_{q_1}(i,L_1)=\false$ and $\ts(v')-\ts(u)\geq M$.
        From Claim~\ref{claim:real-abstr} we also have $\ts_1(v')-\ts_1(u)\geq M$.
        We deduce that $I.up=\infty$ and $\ts(v)-\ts(u)\in I$.
        
        \item Or $\ACC_{q_1}(i,L)=\true$ and 
        $\ts(v')-\ts(u)=\ts_1(v')-\ts_1(u)$. Therefore, 
        \begin{align*}
          \ts(v)-\ts(u) & = \ts(v)-\ts(v') + \ts(v')-\ts(u) \\
          & = \ts_1(v)-\ts_1(v') + \ts_1(v')-\ts_1(u)
          = \ts_1(v)-\ts_1(u)\in I
          \qedhere
        \end{align*}
      \end{itemize}
    \end{itemize}
  \end{itemize}
\end{proof}

\begin{figure}[h]
\centerline{\includegraphics[scale=0.3,page=37]{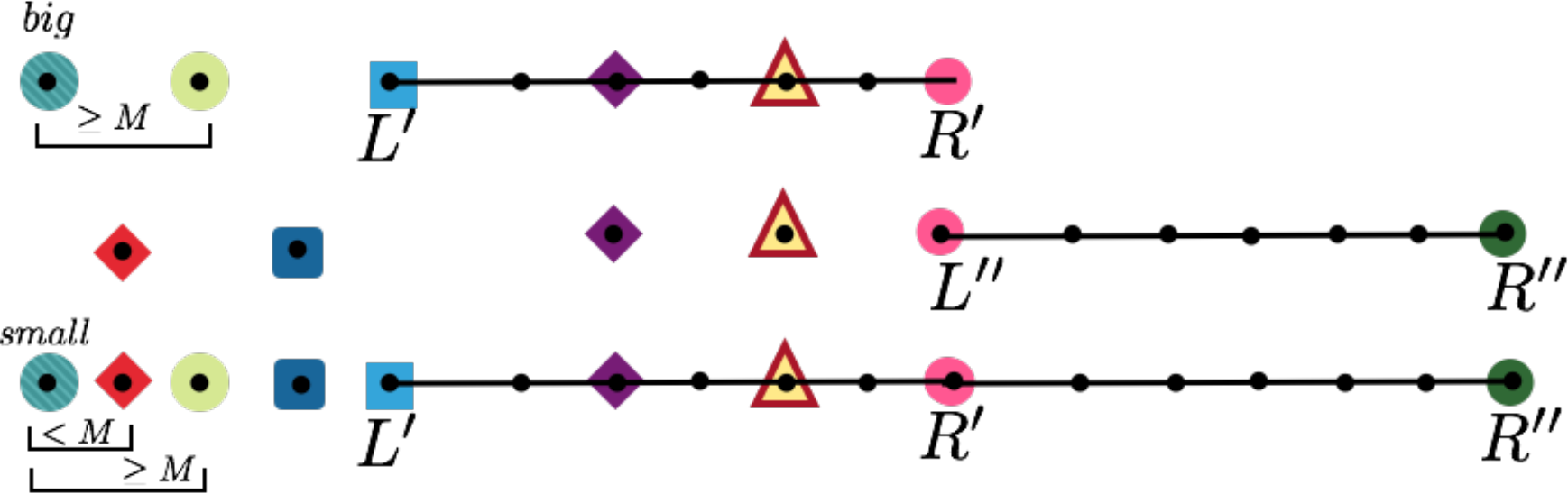}}	
\caption{While doing $\oplus$, the ``accuracy'' of a point $i<L'$ or $i<L''$ can
change from false to true, depending on the new next point obtained after the
combine.  However, if $i$ was accurate before the combine, it will stay
accurate after the combine}
 \label{fig:small-big}
\end{figure}

\begin{claim}\label{claim:oplus}
  Let $i,j\in P_1$ with $i\leq j$ and let $u=\chi(i)$ and $v=\chi(j)$.
  \begin{enumerate}
    \item If $\ACC_{q_1}(i,j)=\false$ then $\ts(v)-\ts(u)\geq M$.
    
    \item If $\ACC_{q_1}(i,j)=\true$ then $\ts(v)-\ts(u)=\ts_1(v)-\ts_1(u)$.
  \end{enumerate}
\end{claim}

\begin{proof}
  The proof is by induction on the number of points in $P_1$ between $i$
  and $j$.  The result is clear if $i=j$.  So assume that
  $k=\nxt_{P_1}(i)\leq j$ and let $w=\chi(k)$.  By induction, the claim
  holds for the pair $(k,j)$.
  \begin{enumerate}
    \item If $\ACC_{q_1}(i,j)=\false$ then either $\acc_1(i)=\false$ or 
    $\ACC_{q_1}(k,j)=\false$. 
    
    In the first case, by definition of the transition for $\oplus$, we
    have either $\ACC_q(i,k)=\false$ or $D_q(i,k)\geq M$.  In both cases,
    we get $\ts(w)-\ts(u)\geq M$ by Claim~\ref{claim:real-abstr}.
    
    In the second case,  we get $\ts(v)-\ts(w)\geq M$ by induction.
    
    Since $\ts$ is non-decreasing, we obtain $\ts(v)-\ts(u)\geq M$.
    
    \item If $\ACC_{q_1}(i,j)=\true$ then $\acc_1(i)=\true$ and 
    $\ACC_{q_1}(k,j)=\true$. 
    
    By induction, we obtain $\ts(v)-\ts(w)=\ts_1(v)-\ts_1(w)$.
    
    From the definition of the transition for $\oplus$, since 
    $\acc_1(i)=\true$, we get $\ACC_q(i,k)=\true$ and $D_q(i,k)<M$. 
    Using Claim~\ref{claim:real-abstr} we deduce that 
    $\ts(w)-\ts(u)=D_q(i,k)$. Now, $D_q(i,k)<M$ implies 
    $D_q(i,k)=d_q(i,k)=d_{q_1}(i,k)$. Using again 
    Claim~\ref{claim:real-abstr} we get 
    $d_{q_1}(i,k)=\ts_1(w)-\ts_1(u)$. We conclude that             
    $\ts(w)-\ts(u)=\ts_1(w)-\ts_1(u)$.
    
    Combining the two equalities, we obtain 
    $\ts(v)-\ts(u)=\ts_1(v)-\ts_1(u)$ as desired.
    \qedhere
  \end{enumerate}
\end{proof}

\subsection{Correctness of \AKMV, Complexity}
\label{app:valid-correct}

\paragraph*{Correctness of the Construction}
\medskip\noindent$(\subseteq)$ 
Let $\tau$ be a \TT accepted by \AKMV. There is an accepting run of \AKMV
reading $\tau$ and reaching state $q$ at the root of $\tau$.  By 
Lemma~\ref{lem:AKMV-tcw}, the term $\tau$ is good and state $q$ is a
realizable abstraction of $\tau$, hence $(\sem{\tau},\hole)$ is a split-\TCW.
But since $q$ is accepting, we have ${\hole}=\emptyset$.  Hence $\sem{\tau}$ is
a \TCW. From \eqref{prop:A4} we deduce that $\sem{\tau}$ is 
realizable and the endpoints of $\sem{\tau}$ are the only colored points by 
\eqref{prop:A1} and the acceptance condition.

\medskip\noindent$(\supseteq)$ 
Let $\tau$ be a good \kTTs such that $\sem{\tau}=(G,\chi)$ is a realizable
\TCW and the endpoints of $\sem{\tau}$ are the only colored points.  
Let $\ts\colon V\to\N$ be a timestamp map satisfying all the timing constraints
in $\tau$.  We construct a run of $\AKMV$ on $\tau$ by resolving the
non-deterministic choices as explained below.  Notice that the transitions for
$\rename{i}{j}$ and $\forget{i}$ are deterministic.
We will obtain an accepting run $\rho$ of \AKMV on $\tau$ such that for every
subterm $\tau'$, the state $\rho(\tau')$ satisfies \eqref{prop:A4} with
timestamp map $\ts$, or more precisely, with the restriction of $\ts$ to the
vertices in $\sem{\tau'}$.

\begin{itemize}
  \item A leaf $(a,i)\procrel(b,j)$ of the term $\tau$ corresponds to two 
  vertices $u,v\in V$ with $u\procrel v$. We have $i<j$ since $\tau$ is good so 
  the transition is enabled for this atomic subterm.
  We resolve non-determinism by setting
  $\tsm(i)=\ts(u)[M]$, $\tsm(j)=\ts(v)[M]$ and $\acc(i)=\true$ iff
  $\ts(v)-\ts(u)<M$.
  Therefore, \eqref{prop:A4} holds with $\ts$.

  \item A leaf $(a,i)\matchrel^{I}(b,j)$ of the term $\tau$ corresponds to two 
  vertices $u,v\in V$ with $u\matchrel^{I} v$. Since $\ts$ satisfies all timing 
  constraints, we have $\ts(v)-\ts(u)\in I$. 
  The transition taken at this leaf resolves non-determinism by setting
  $\tsm(i)=\ts(u)[M]$, $\tsm(j)=\ts(v)[M]$ and $\acc(i)=\true$ iff
  $\ts(v)-\ts(u)<M$. We can check that all conditions enabling this transition
  are satisfied. Moreover, \eqref{prop:A4} holds with $\ts$.

  \item We can check that the conditions enabling transitions at
  $\rename{i}{j}$ or $\forget{i}$ nodes are satisfied
  since $\tau$ is good and $\sem{\tau}$ is a \TCW whose endpoints are
  colored.
  
  \item Consider a subterm $\tau'=\tau_1\oplus\tau_2$.  Let
  $\rho(\tau_1)=q_1=(P_1,L_1,\tsm_1,\acc_1)$ and
  $\rho(\tau_2)=q_2=(P_2,L_2,\tsm_2,\acc_2)$. 
  Define $q'=(P',L',\tsm',\acc')$ by $P'=P_1\cup P_2$, $L'=L_1$,
  $\tsm'=\tsm_1\cup\tsm_2$ and for all $i\in P'$, $\acc'(i)=\true$ iff 
  $i^{+}\neq\infty$ and $\ts(\chi'(i^{+}))-\ts(\chi'(i))<M$.
  We show that $q_1,q_2\xrightarrow{\oplus}q'$ is a transition. 
  
  The condition $R_1=\max(P_1)=L_2$ and $\{i\in P_2\mid L_1\leq i\leq
  R_1\}\subseteq P_1$ holds since $\tau$ is a good term and $q_1$, $q_2$ are
  realizable abstractions of $\tau_1$, $\tau_2$.
  
  Now, we look at the condition on $\acc'$.  Let $i\in
  P_1\setminus\{\max(P_1)\}$ and $j=\nxt_{P_1}(i)$.  We have $\acc_1(i)=\true$
  iff $\ts(\chi_1(j))-\ts(\chi_1(i))<M$ since \eqref{prop:A4} holds with $\ts$
  at $\tau_1$.  The latter holds iff for all $k\in P'$ with $i\leq k<j$ we have
  $\ts(\chi'(k^{+}))-\ts(\chi'(k))< M$ (i.e., $\ACC_{q'}(i,j)=\true$ by the
  above definition of $\acc'$) and $D_{q'}(i,j)< M$ (again, by the definition of
  $\acc'$ we have that $\acc'(k)=\true$ implies
  $d_{q'}(k,k^{+})=\ts(\chi'(k^{+}))-\ts(\chi'(k))$ and $\ACC'(i,j)=\true$
  implies $D_{q'}(i,j)=\ts(\chi'(j))-\ts(\chi'(i))$).  
\end{itemize}

\subsection*{Complexity of \AKMV} A state of \AKMV has the form 
$(P, L, \tsm, \acc)$ where $P$ is a subset of $K$, and \tsm, \acc are maps 
from $P$. Clearly, the complexity is dominated by the map \tsm as long as 
$M \geq 2$. Thus, the 
 number of states of \AKMV 
is $M^{\mathcal{O}(K)}$. 
\end{proof}

\section{Tree Automaton for the Timed System}
\label{app:sys-trans}
In this section, we give the full list of transitions of the tree automaton \AKMS.
The transitions of \AKMS are described in Tables~\ref{tab:AKMS-1} and \ref{tab:AKMS-2}. 
\begin{table}[h]
  \noindent\hspace{-5mm}
  \begin{tabular}{|c|p{120mm}|}
    \hline
    $(a,i)\procrel(b,j)$ 
    &
    $\xrightarrow{(a,i)\procrel(b,j)}q$ 
    is a transition if $i<j$ and $P=\{i,j\}$, $L=i$, $\pushb=\popb=0$, $G_x=0$
    and $Z_x=\bot$ for all $x\in X$, $\delta$ is guessed such that
    $\mathsf{target}(\delta(i))=\mathsf{source}(\delta(j))$.
    Further, if $a=\epsilon$ then we take a special initial dummy transition
    $\delta(i)=(s_\mathsf{dummy},\true,\varepsilon,\nop,X,s_0)$.
    \\ \hline
    $\rename{i}{j}$
    &
    $q\xrightarrow{\rename{i}{j}}q'$ is a transition if $i\in P$ and
    $i^-<j<i^+$, and $q'$ is obtained from $q$ by replacing $i$ by $j$.
    \\ \hline
    $\forget{i}$
    &
    $q\xrightarrow{\forget{i}}q'$ is a transition if $i\in P$, $L<i<\max(P)$
    (endpoints should stay colored) and for each $x\in\reset(i)$ there exists
    $j\in P$ such that $i<j\leq R$ and $x\in\reset(j)$ (the last reset point of 
    each clock is never forgotten, even if it is an internal point).
    Then, state $q'$ is deterministically given by $P'=P\setminus\{i\}$, $L'=L$,
    $\delta'=\delta_{\mid P'}$, $\pushb'=\pushb$, $\popb'=\popb$, $G'=G$ and
    $Z'=Z$.
    \\ \hline
    $(a,i)\matchrel^{I}(b,j)$ 
    &
    $q,(a,i)\matchrel^{I}(b,j)\xrightarrow{\oplus}q'$ 
    is a transition if  one of the following conditions holds:

    $\mathsf{M}_1$: 
      the automaton guesses that it is a stack edge:
      $i=L<j=R=\max(P)$ and $\pushb=0=\popb$ (a push-pop edge may be added to
      $(L,R)$ only if no push-pop edges were added to $L$ or to $R$ before),
      $\delta(i)$ contains some $\push_c$ operation and $\delta(j)$ contains a
      $\pop_c^{I}$ operation.  Then, $P'=P, L'=L, \delta'=\delta$,
      $\pushb'=1=\popb'$, $G'=G$ and $Z'=Z$.

      \includegraphics[scale=0.3,page=26]{all-figs} 
  
      $\mathsf{M}_2$: 
      Or, the automaton guesses that it is a constraint for some clock $x\in X$:
      \begin{itemize}[nosep]
        \item[(a)] $i<j=R=\max(P)$, $G_x=0$, $x\in I$ is in $\delta(j)$ and $G'_x=1$, 

        \item[(b)] ($i\in P$ or $i<L$) and $x\in\reset'(i)$ and $x\not\in\reset(k)$ for
        all $k\in P$ with $i<k<j$

        \item[(c)] Either ($L\leq i$ and $Z'_x=Z_x$) or ($i<L$ and $Z_x\in
        \{\bot,i\}$ and $Z'_x=i$).
      \end{itemize}
      Then, $P'=P\cup \{i\}$, $L'=L$, $\delta'(k)=\delta(k)$ for all $k\in P$,
      $\pushb'=\pushb$, $\popb'=\popb$, $G'_y=G_y$ and $Z'_y=Z_y$ for all $y\in
      X\setminus\{x\}$.  Note that if $i\not\in P$, then $\delta'(i)$ is
      guessed. The figure below considers the case when $i < L, i \notin P$, that is, $i$ 
      is a new reset point. This $i$ gets added in the set of active colors $P'$. 
       \includegraphics[scale=0.35,page=25]{all-figs} 
  
  \\
     \hline
   \end{tabular}
  \caption{Transitions of \AKMS. 
  $q=(P,L,\delta,\pushb,\popb,G,Z)$, $q'=(P',L',\delta',\pushb',\popb',G',Z')$.}
  \protect\label{tab:AKMS-1}
\end{table}

\begin{table}[h!]
  \noindent\hspace{-2mm}
  \begin{tabular}{|c|p{130mm}|}
    \hline
    $\oplus$
    &
    $q',q''\xrightarrow{\oplus}q$ is a transition if
    the following hold:
    
             \centerline{\includegraphics[scale=0.3,page=27]{all-figs}}
     \begin{enumerate}[nosep,label= $\mathsf{C}_{\arabic*}$:,ref=$\mathsf{C}_{\arabic*}$]
      \item $R'=\max(P')=L''$ and $\{i\in P''\mid L'\leq i\leq R'\}\subseteq
      P'$ (we cannot insert a new point from the second argument in the
      $\procrel$-block of the first argument).
 Note that according to $\mathsf{C_1}$, the points
\protect\includegraphics[scale=0.25,page=28]{all-figs},
\protect\includegraphics[scale=0.25,page=29]{all-figs} and
\protect\includegraphics[scale=0.25,page=35]{all-figs} in $P''$ lying between
$L', R'$ are already points in the non-trivial block connecting $L'$ to $R'$.

      \item $\forall i\in P'\cap P''$, $\delta'(i)=\delta''(i)$ (the guessed
      transitions match).
      By $\mathsf{C}_2$, the transitions $\delta', \delta''$ of
      \protect\includegraphics[scale=0.25,page=30]{all-figs},
      \protect\includegraphics[scale=0.25,page=31]{all-figs},
      \protect\includegraphics[scale=0.25,page=36]{all-figs},
      \protect\includegraphics[scale=0.25,page=28]{all-figs},
      \protect\includegraphics[scale=0.25,page=29]{all-figs} and
      \protect\includegraphics[scale=0.25,page=35]{all-figs} must match.

      \item if there is a $\pushb$ operation in $\delta''(L'')$ then
      $\pushb''=1$ and if there is a pop operation in $\delta'(R')$ then
      $\popb'=1$ (the push-pop edges corresponding to the merging point have
      been added, if they exist). By $\mathsf{C}_3$, if $\delta(R')=\delta(L'')$ contains a pop (resp.\ push)
operation then $R'=L''$ is the target (resp.\ source) of a push-pop edge.

      \item if some guard $x\in I$ is in $\delta(R')$, then $G'_x=1$ 
      (before we merge, we ensure that the clock guard for $x$ in the transition guessed at $R'$, if any, has been checked). After the merge, $R'=L''$ becomes an internal point; hence by $\mathsf{C}_4$,
any guard $x \in I$ in $\delta'(R')$ must be checked already, i.e., $G'_x=1$.
After the merge, it is no more possible to add an edge $\curvearrowright^I$
leading into $R'$.

      \item if $Z'_x\neq\bot$, then $\forall j\in P''$, $Z'_x<j< L'$ implies
      $x\not\in Reset''(j)$ (If a matching edge starting at $Z'_x<L'$ had been
      seen earlier in run leading to $q'$, then $x$ should not have been reset
      in $q''$ between $Z'_x$ and $L'$, else it would violate the consistency of clocks).  By
      $\mathsf{C}_5$, if $Z'_x$ is
\protect\includegraphics[scale=0.25,page=30]{all-figs} (resp.\
\protect\includegraphics[scale=0.25,page=31]{all-figs}), i.e.,
\protect\includegraphics[scale=0.25,page=30]{all-figs} (resp.\
\protect\includegraphics[scale=0.25,page=31]{all-figs}) is the source of a
timing constraint $\curvearrowright^I$ for clock $x$ whose target is in the
$L'$--$R'$ block,
then clock $x$ cannot be reset at
\protect\includegraphics[scale=0.25,page=32]{all-figs} and 
\protect\includegraphics[scale=0.25,page=33]{all-figs}
(resp.\ \protect\includegraphics[scale=0.25,page=33]{all-figs}).

      \item if $Z''_x\neq\bot$, then $\forall j\in P'$, $Z''_x<j<L''$ implies
      $x\not\in Reset'(j)$ (If a matching edge starting at $Z''_x<L''$ had been
      seen earlier in run leading to $q''$, then $x$ should not have been reset
      in $q'$ between $Z''_x$ and $L''$).  By $\mathsf{C}_6$, if $Z''_x$ is
\protect\includegraphics[scale=0.25,page=32]{all-figs}, then $x$ cannot be reset
at \protect\includegraphics[scale=0.25,page=31]{all-figs}, 
\protect\includegraphics[scale=0.25,page=36]{all-figs},
\protect\includegraphics[scale=0.25,page=28]{all-figs}, or
\protect\includegraphics[scale=0.25,page=29]{all-figs}.
Likewise, if $Z''_x$ was \protect\includegraphics[scale=0.25,page=33]{all-figs},
then clock $x$ cannot be reset at
\protect\includegraphics[scale=0.25,page=36]{all-figs},
\protect\includegraphics[scale=0.25,page=28]{all-figs}, or
\protect\includegraphics[scale=0.25,page=29]{all-figs}.
      
      \item $P=P'\cup P''$, $L=L'$, $\delta=\delta'\cup \delta''$,
      $\pushb=\pushb'$, $\popb=\popb''$, $G=G''$ and \newline
      for all $x\in X$ we have $Z_x=Z''_x$ if $Z''_x<L'$, else $Z_x=Z'_x$. 
      $\mathsf{C}_7$ says that on merging, we obtain
the third split-\TCW.
After the merge, if $Z_x$ is defined, it must be on the left of $L'$, i.e., one of
\protect\includegraphics[scale=0.25,page=30]{all-figs},
\protect\includegraphics[scale=0.25,page=32]{all-figs},
\protect\includegraphics[scale=0.25,page=31]{all-figs},
\protect\includegraphics[scale=0.25,page=33]{all-figs}.

      Notice that the above three conditions ensure the well-nestedness of clocks. 
  By $\mathsf{C}_5$ and $\mathsf{C}_6$ we cannot have both
$Z'_x\in\{\protect\includegraphics[scale=0.25,page=30]{all-figs},
\protect\includegraphics[scale=0.25,page=31]{all-figs}\}$ and
$Z''_x\in\{\protect\includegraphics[scale=0.25,page=32]{all-figs},
\protect\includegraphics[scale=0.25,page=33]{all-figs}\}$.  So if
$Z''_x\in\{\protect\includegraphics[scale=0.25,page=32]{all-figs},
\protect\includegraphics[scale=0.25,page=33]{all-figs}\}$ then $Z_x=Z''_x$ and
otherwise $Z_x=Z'_x$ (including when
$Z''_x\in\{\protect\includegraphics[scale=0.25,page=28]{all-figs}, 
\protect\includegraphics[scale=0.25,page=29]{all-figs}\}$ and $Z'_x=\bot$).
  
      \end{enumerate}\\
     \hline
   \end{tabular}
  \caption{$\oplus$ transitions of \AKMS. 
  $q=(P,L,\delta,\pushb,\popb,G,Z)$, $q'=(P',L',\delta',\pushb',\popb',G',Z')$,
  $q''=(P'',L'',\delta'',\pushb'',\popb'',G'',Z'')$.}
  \protect\label{tab:AKMS-2}
\end{table}

\subsection{Proof of Theorem \ref{thm:AKMS}}
\label{app:AKMS}
\begin{proof}[Proof sketch]
  Let $\stt$ be a \kTT. We will show that $\stt$ is accepted by \AKMS iff $\stt$
  is restricted and $\sem{\stt}\in\TCW(\Sys)$.

  Assume that \AKMS has an accepting run on $\stt$.  Clearly, $\stt$ is
  restricted.  Now, the first two components $(P,L)$ of \AKMS behave as the
  corresponding ones in \AKMV and ensure that $\tcw=\sem{\stt}$ is indeed a \TCW
  (which need not be realizable). It remains to check that \Sys admits a run on 
  \tcw.
  
  \begin{enumerate}
    \item We first define the sequence of transitions.  Each vertex $v$ of \tcw
    is introduced as node colored $j$ in (1) some atomic term $i\procrel j$ if
    $v$ is not minimal in \tcw or in (2) some atomic term $j\procrel k$ if $v$
    is not maximal in \tcw.  We let $\delta(v)=\delta(j)$ be the transition
    guessed by \AKMS when reading this atomic term.  Notice that if both cases
    above occur, i.e., if $v$ is internal, then \AKMS has to guess the same
    transition by $\mathsf{C}_2$.  Notice also that if $u\procrel v$ in \tcw
    then for some atomic term $i\procrel j$ occurring in $\stt$ we have
    $\delta(u)=\delta(i)$ and $\delta(v)=\delta(j)$.  Therefore,
    $\mathsf{target}(\delta(u))=\mathsf{source}(\delta(v))$.  So we have
    constructed a sequence of transitions $(\delta(v))_v$ which forms a path in
    \Sys reading \tcw.  By the acceptance condition, if $v$ is the minimal
    (resp.\ maximal) vertex of \tcw then $\delta(v)$ is the initial dummy
    transition (resp.\ $\mathsf{target}(\delta(v))$ is final).

    \item To ensure that the TCW is generated by the system, we check that the
    sequence of push-pop operations is well-nested.  This is achieved using bits
    $\pushb$, $\popb$ and the following facts (i) by $\mathsf{C}_3$ and the
    accepting condition, every vertex $v$ such that $\delta(v)$ contains a push
    (resp.\ pop) operation is the source (resp.\ target) of a matching push-pop
    edge, (ii) push-pop edges are only within the non-trivial block, (iii) the
    left end-point is the source of a push-pop edge iff $\pushb=1$ and the right
    end-point is the target of a push-pop edge iff $\popb=1$, (iv) by
    $\mathsf{M}_1$ a push-pop edge is added only between the left and right end
    points of the non-trivial block and only when $\pushb=0=\popb$, which are
    updated to $1$, (v) a combine fuses the right end of a non-trivial block
    with the left end of another one.  Because the transition of this fused
    point cannot carry both a push operation and a pop operation, this ensures
    well-nesting under the combine operation.

    \item The last subtle point concerns clock constraints.  First,
    $\mathsf{M}_2$, $\mathsf{C}_4$ and the acceptance condition, with the help
    of $G_x$ make sure that if the transition $\delta(v)$ of some point $v$
    contains a guard $x\in I$ then $v$ is the target of a $\matchrel^{I}$ edge
    whose source is some node $u$ such that $x$ is reset in $\delta(u)$.  Let us
    explain why $u$ is the last reset for clock $x$ in the past of $v$.  By
    $\mathsf{M}_2$ when the edge $u\matchrel^{I}v$ is added, $v$ is the current
    right most point and no colored points between $u$ and $v$ reset clock $x$.
    Since the last reset of clock $x$ is never forgotten (see transition for
    $\forget{i}$), we deduce that there are no resets of clock $x$ between $u$
    and $v$ in the current split-\TCW. If $u$ is in the non-trivial block, no
    further points will be added between $u$ and $v$, hence we are done.  Assume
    now that $u$ is a reset point colored $i<L$ on the left of the non-trivial
    block in the current split-\TCW. By $\mathsf{M}_2$ we store the color $i$ of
    $u$ in $Z_x$.  When we later use a combine operation, $\mathsf{C}_5$ and
    $\mathsf{C}_6$ ensure that no transitions resetting clock $x$ are inserted
    between $u$ and (the non-trivial block containing) $v$.
    
  \end{enumerate}
  Thus, we obtain that $\tcw$ is indeed generated by $\Sys$, i.e., $\tcw\in
  \TCW(\Sys)$.  In the reverse direction, if $\tcw\in \TCW(\Sys)$, then there is
  a sequence of transitions which lead to the accepting state on reading $\tcw$.
  By guessing each of these transitions correctly at every point, we can
  generate the run of our automaton \AKMS. 
\end{proof}

\subsection*{Complexity of \AKMS}
 Recall that a state of \AKMS is a tuple $q=(P,L,\delta,\pushb,\popb,G,Z)$ where,
\begin{itemize}
  \item $P$ is the set of active colors, and $L=\mathsf{Left} \in P$ is the left-most point
  that is connected to the right-end-point $R=\mathsf{Right}=\max(P)$ by successor edges on the non-trivial block. 
  $P$ thus any subset of $K$.

  \item $\delta$ is a map that assigns to each color $k\in P$ the transition $\delta(k)$ guessed at the leaf corresponding to color $k$,
$\delta$ has size $|\Sys|^{\mathcal{O}(K)}$ where $|\Sys|$ denotes the size (number of transitions)  of the TPDA,  

  \item $\pushb$ and $\popb$ are two boolean variables: $\pushb=1$ iff a push-pop edge has been added to $L$ and $\popb=1$ iff a push-pop edge has been added to $R$,
  
  \item $G=(G_x)_{x\in X}$ is a boolean vector of size $|X|$: for each clock $x\in X$, $G_x=1$ iff some constraint on $x$ has already been checked at $R$. The number of possible vectors is thus $2^{\mathcal{O}(|X|)}$, 

  \item $Z=(Z_x)_{x\in X}$ assigns to each clock $x$ either the color $i\in P$
  with $i<L$ of the unique point on the left of the non-trivial block which is
  the source of a timing constraint $i\matchrel^{I}j$ for clock $x$, or $\bot$
  if no such points exist. The size of $Z$ is thus $(K+1)^{\mathcal{O}(|X|)}$.
\end{itemize}
Clearly, the number of states of \AKMS is $\leq |\Sys|^{\mathcal{O}(K)}(K+1)^{\mathcal{O}(|X|)}$. 

\subsection{Proof of Theorem \ref{thm:sys-cor}}
\label{app:sys-corr}

\begin{proof}[Proof sketch]
  ($\implies)$ If $L(\Sys)$ is not empty, then there exists a realiziable TCW
  $W$ accepted by $\Sys$.  Now $W$ is well-timed and hence we know that its
  tree-width is bounded by a constant $K\leq 3|X|+3$.  That is, by the proof of
  Lemma~\ref{TPDA-bound} in Section~\ref{sec:stw-tpda}, Eve has a
  winning-strategy on $W$ with at most $K$ colors.  Further, we may observe that
  Eve's strategy on $W$ gives us a $K$-\TT $\tau$ which a \emph{good} and, in
  fact, a \emph{restricted} $K$-\TT, such that $\sem{\tau}=W$.

  Now, $\tau$ is a good $K$-\TT such that $\sem{\tau}=W$ is a realizable \TCW.
  Thus, by Theorem~\ref{thm:AKMV}, $\tau\in L(\AKMV)$.  Further, as $\tau$ is
  restricted and $W\in L(\Sys)$, by Theorem~\ref{thm:AKMS}, $\tau\in L(\AKMS)$.
  Thus we have have $L(\AKMS\cap\AKMS)\neq \emptyset$.

  $(\Longleftarrow)$ Let $\tau\in L(\AKMS\cap \AKMS)$.  Then, by
  Theorem~\ref{thm:AKMS}, we get that $\sem{\tau}\in\TCW(\Sys)$.  Again by
  Theorem~\ref{thm:AKMV}, $\sem{\tau}$ is a realizable $\TCW$.  Thus, we get
  that $\sem{\tau}\in L(\Sys)$.
\end{proof}

\section{Implmententation and Experimental Results}
\label{app:exp}

In this section, we give the  missing details regarding constraints 
of the maze in Figure \ref{tpda-maze-tricky-egx}, rules for 
constructing the TPDA from the maze, and finally 
the performance of our tool on some examples.

\subsection{Detailed constraints for Figure \ref{tpda-maze-tricky-egx}}
For convenience, we reproduce the figure here. 
 \begin{figure}[t]
\begin{center}
\includegraphics[height=5cm,width=6cm]{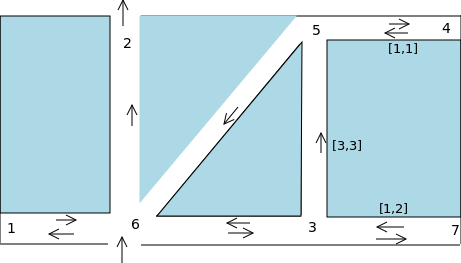}
\includegraphics[scale=0.3]{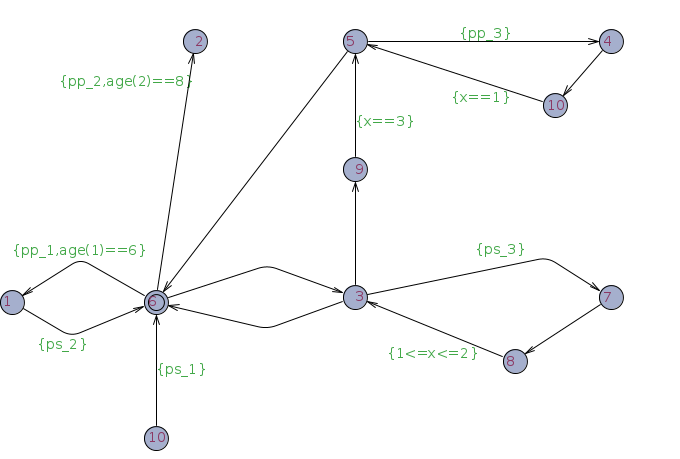}
 \end{center}
\caption{A simple maze. Every junction, dead end, entry point or exit point is called a  place. 
The places are 1 to 7. 6 is the entry and 2 is the exit.  1, 7 and 4 are dead ends. 
Time intervals denote the time taken between adjacent places; for instance, a time between 1 and 2 units 
must be elapsed to go between places 3,7. A unidirectional arrow represents a one way, for example, places 3 to 5. 
On the right, is the TPDA model of the maze.  
} 
\label{app:tpda-maze-tricky-egx}
\end{figure}

 The following are the constraints that must be respected to traverse the maze 
 and successfully exit it, starting at the entry point. 
\begin{itemize}
\item[Type (1)] Logical constraints: Place 1 must be visited exactly once. From the time we enter the maze to the visiting of place 1, one must visit place 7 and place 4 equal number of times and at any point during time, number of visits to place 7 is not less than number of visits to place 4. Similarly, from visiting place 1 to exiting the maze, one must visit place 7 and place 4 equal number of times and again at any point during time, number of visits to place 7 is not less than number of visits to place 4. (thus, we model loading at $7$ and unloading at $4$).

\item[Type (2)] Local time constraints: Time taken from one place to another adjacent place is as given in the Figure. The time taken from some place $i$ to another adjacent place $j$ is given as a closed interval $[a_{ij},b_{ij}]$ along with the arrow. $a_{ij}$ is the least time taken from place $i$ to $j$ and $b_{ij}$ is the upper bound on time taken from place $i$ to place $j$. One cannot spend any time between a pair of adjacent places other than the ones specified in the maze. For example,   we have not also specified time taken from place 6 to place 1. So time bound for going from place 6 to 1 is $[0,0]$. Further, one cannot stay in any place for non-zero duration.
\item[Type (3)] Global time constraints: From entering into the maze to visiting of place 1, time taken should be exactly $m$ units (a parameter).  From visiting place 1 to exit, time bound should be exactly $n$ units (another parameter).
\end{itemize}

\subsection{Maze to TPDA construction details}
	\textbf{ States of the automaton} :
	For each place of the maze we have a corresponding state in the automaton. We call these states as regular states. For each of the constraints, we may have to add extra states in the automaton. For the constraints of type (2), if there is a time bound other than zero([0, 0]) from place $i$ to place $j$, then we add an extra state $k$ in the automaton between regular state $i$ and regular state $j$. If it has been given time bound from entry of the maze to visiting of some place, then we have to add one more extra state before the entry point.

\textbf{Transitions of the automaton} :
	If two states $i$ and $j$ are adjacent, then there is a transition $T_{ij}$ for this. By default, we can't stay non-zero time in a regular state $i$. So, clock $x_1$ is reset in the incoming transitions to state $i$ and there is a check $x_1==0$ in the outgoing transitions from state $i$ for each state $i$. If we have to stay $[a_i, b_i]$ time in the state i, then the check could have been $a_i <= x_1 <= b_i$. If there is a time bound $[a,b]$ going from place $p$ to some of its adjacent place $q$(Type (2) constraint) and if place $r$ is added in between, then one transition added from place $p$ to $r$, where clock $x$ reset and one transtion from place $r$ to $q$ added, where check of clock x happens with constraint $[a,b]$. Nested time bound can be done using pushing and popping of same symbol between two events. If it has been given that time bound from entering into the maze to visiting of place $p$ must be $[a,b]$(Type (3) constraint), then push some symbol while entering into the maze and pop the same symbol while visiting place $p$ such that age of the stack symbol must belong to $[a,b]$. If one has to visit place $p_1$ and place $p_2$ same number of times(Type (1) Constraint), then it can be done by pushing some symbol while visiting place $p_1$ and popping of the same symbol while visiting place $p_2$.

\subsection{Experimental Results}

\label{app:expt}

\subsubsection{Constraints on maze 2}
\begin{itemize}
\item[Type (1)] Place 1 and place 2 must be visited exactly once.

\item[Type (2)] Time taken from one place to another adjacent place is given in the maze itself. One cannot spend any time between a pair of adjacent places other than the ones specified in the maze. Further, one cannot stay in any place for non-zero duration. For any two adjacent places $p$ and $q$, one can go from p to q or q to p. In other words all corridors in the maze are bidirectional.

\item[Type (3)]
Total time taken to visit the maze should be at least 5 time units and at most 7 time units.

\end{itemize} 

\begin{figure}[h!]
\includegraphics[scale=0.3]{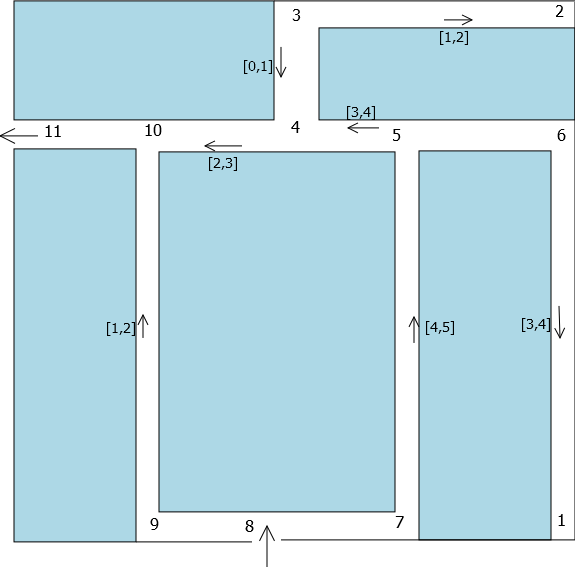}
\includegraphics[scale=.43]{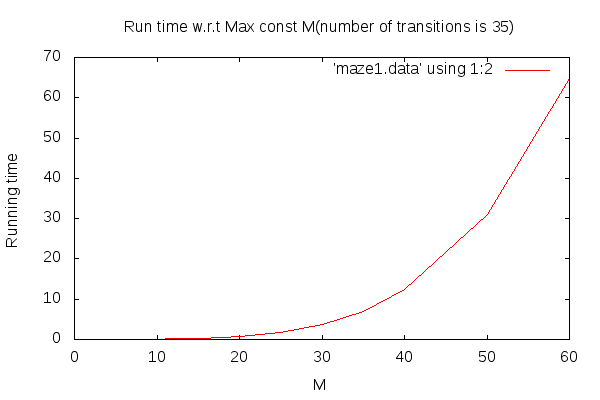}
\caption{Maze 2 and its corresponding  plot}
\label{fig:graph2}
\end{figure}

\subsubsection{Constraints on maze 3}
\begin{itemize}
\item[Type (1)] You must visit place 1, place 2 and place 3 exactly once. Between entering into the maze and visiting of place 1, one must visit place 4 and place 5 same number of times, but in any moment number of visits to place 4 is not less than the number of visits to place 5. Same type of constraints on place 4 and 5 applied between visiting of place 2 and visiting of place 3. Again Same type of constraints on place 4 and 5 applied between visiting of place 3 and exiting from the maze.
\item[Type (2)] Time spent between two adjacent places is given on the maze itself. You have to stay in place 1 and place 2 exactly one time unit for each of them. You can't stay in other places except 1 and 2. All corridors in the maze are bidirectional.

\item[Type (3)]  Total time taken to visit the maze is [9,9]. After entering into the maze, one must visit place 1 within [1,3]. After visiting place 1, one must visit place 2 within [5,5]. After visiting place 2, one must exit from the maze in [3,5] time.
\end{itemize}

\begin{figure}[h!]
\includegraphics[scale=0.3]{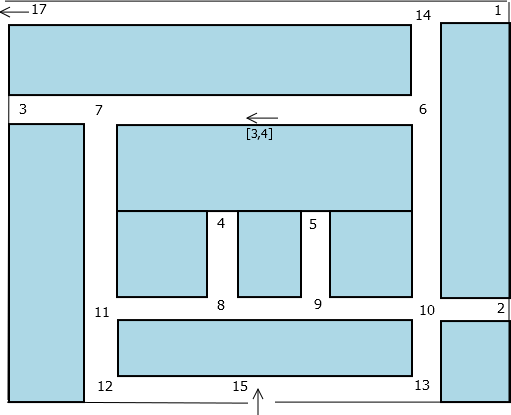}
\includegraphics[scale=0.43]{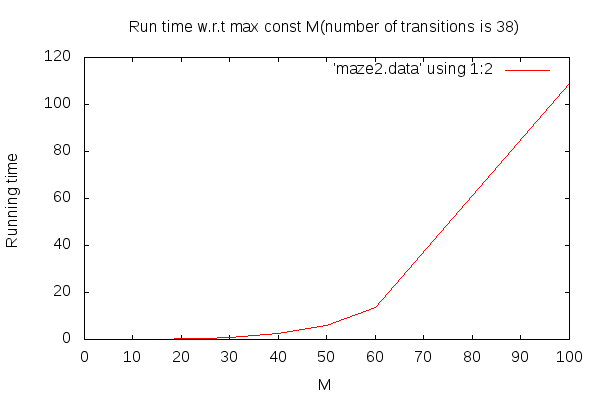}
\caption{Maze 3 and its corresponding  plot }

\end{figure}

\subsubsection{Constraints on maze 4}
\begin{itemize}
\item[Type (1)] You must visit place 1 and place 2 exactly once.

\item[Type (2)]  You can't spend any time between two places except the ones specified in the maze itself. One can't stay in any places for non-zero time except for place 4, where one can stay for [1,2] time unit. All corridors in the maze are bidirectional.

\item[Type (3)] Global time bound or total time taken to visit the maze is [3,4]. After visiting place 1, one must visit place 2 within [3,3] time. After visiting place 2, one must exit from the maze within [2,3] time.
\end{itemize}

\begin{figure}[h!]
\includegraphics[scale=0.3]{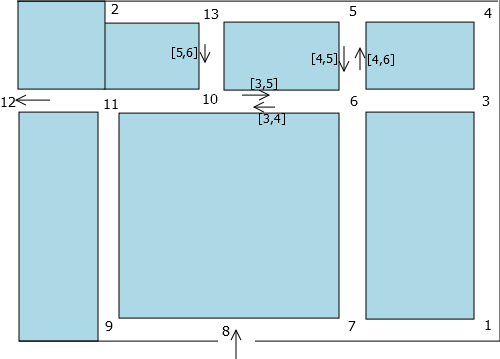}
\includegraphics[scale=0.43]{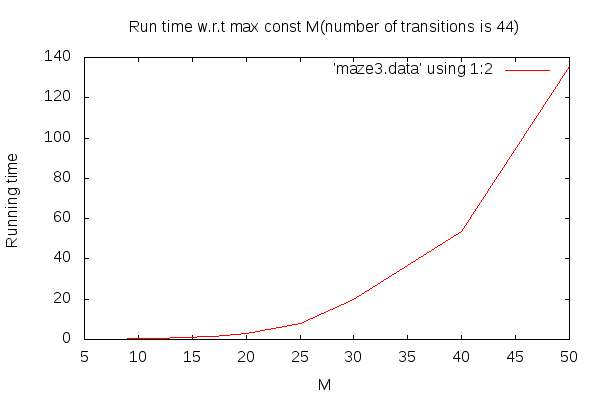}
\caption{Maze 4 and its corresponding plot}

\label{fig:graph-big}
\end{figure}

\end{document}